%% file: paper.tex
\documentclass[11pt]{article}
\pdfoutput=1

\usepackage{fullpage}
\usepackage{graphicx}   
\usepackage{subfig}  
\usepackage{hyperref}   
\usepackage{mdwlist}
\usepackage{xspace}
\usepackage[usenames,dvipsnames]{color}
\usepackage{amsmath, amsthm,amssymb}   

\input{macros}

\usepackage{mathpazo} 
\usepackage[scaled]{helvet} 
\usepackage{courier} 
\normalfont
\usepackage[T1]{fontenc}

\begin{document}

\title{Communication Steps for Parallel Query Processing}

\author{Paul Beame, Paraschos Koutris and Dan Suciu\\
\{beame,pkoutris,suciu\}@cs.washington.edu\\
University of Washington}

\maketitle

\begin{abstract}
We consider the problem of computing a relational query $q$ on a large input database of size $n$, using a large number $p$ of servers. The computation is performed in {\em rounds}, and each server can receive only $O(n/p^{1-\varepsilon})$ bits of data, where $\varepsilon \in [0,1]$ is a parameter that controls replication. We examine how many global communication steps are needed to compute $q$. We establish both lower and upper bounds, in two settings. 
For a single round of communication, we give lower bounds in the strongest possible model, where arbitrary bits may be exchanged; we show that any algorithm requires $\varepsilon \geq 1-1/\tau^*$,
where $\tau^*$ is the fractional vertex cover of the hypergraph of $q$. We also give an algorithm that
matches the lower bound for a specific class of databases.
For multiple rounds of communication, we present lower bounds in a model where routing decisions for a tuple are tuple-based. We show that for the class of {\em tree-like} queries there exists a tradeoff between the number of rounds and the space exponent $\varepsilon$.
The lower bounds for multiple rounds are the first of their kind. Our results also imply that transitive
 closure cannot be computed in $O(1)$ rounds of communication.
\end{abstract}

\input{introduction}

\input{model}

\input{onestep}
\input{multistep}

\input{conclusion}

\bibliographystyle{abbrv}
\bibliography{bib}

\end{document}

%% file: macros.tex
\newcommand{\withoutentropy}[2]{#1}                     

\newcommand{\ifnotentropy}[2]{\withoutentropy{#1}{#2}}

\newcommand{\paris}[1]{{\color{SeaGreen} Paris: [{#1}]}}

\newcommand{\cut}[1]{}

\newenvironment{packed_enum}{
\begin{enumerate}
   \setlength{\itemsep}{1pt}
  \setlength{\parskip}{0pt}
   \setlength{\parsep}{0pt}
}
{\end{enumerate}}

\newcommand{\set}[1]{\{#1\}}                    
\newcommand{\setof}[2]{\{{#1}\mid{#2}\}}        
\usepackage{aliascnt}  		

\newtheorem{theorem}{Theorem}[section]          	
\newaliascnt{lemma}{theorem}				
\newtheorem{lemma}[lemma]{Lemma}              	
\aliascntresetthe{lemma}  					
\newaliascnt{conjecture}{theorem}			
\aliascntresetthe{conjecture}  				
\newaliascnt{remark}{theorem}				
              
\aliascntresetthe{remark}  					
\newaliascnt{fact}{theorem}				
              
\aliascntresetthe{fact}  					
\newaliascnt{corollary}{theorem}			
\newtheorem{corollary}[corollary]{Corollary}      
\aliascntresetthe{corollary}  				
\newaliascnt{definition}{theorem}			
\newtheorem{definition}[definition]{Definition}    
\aliascntresetthe{definition}  				
\newaliascnt{proposition}{theorem}			
\newtheorem{proposition}[proposition]{Proposition}  
\aliascntresetthe{proposition}  				
\newaliascnt{example}{theorem}			
\newtheorem{example}[example]{Example}  	
\aliascntresetthe{example}  				

\newcommand{\ba}[0]{\mathbf{a}}

\newcommand{\mpc}[0]{MPC}

\newcommand{\E}{\mathbf{E}}
\renewcommand{\P}{\mathbf{P}}

\newcommand{\contracted}{{\overline M}}

%% file: introduction.tex
\section{Introduction}
\label{sec:introduction}

  Most of the time spent in big data analysis today is allocated in data
  processing tasks, such as identifying relevant data, cleaning,
  filtering, joining, grouping, transforming, extracting features, and
  evaluating results~\cite{DBLP:conf/pods/Chaudhuri12,datasciencesurvey}.  These
  tasks form the main bottleneck in big data analysis, and a major
  challenge for the database community is improving the performance
  and usability of data processing tools.  The motivation for this
  paper comes from the need to understand the complexity of query
  processing in big data management.

  Query processing is typically performed on a shared-nothing parallel 
  architecture. In this setting, the  data is stored on a large number of 
  independent servers interconnected by a fast network.  The servers 
  perform local computations, then exchange data in global data shuffling 
  steps.  This model of computation has been popularized by 
  MapReduce~\cite{DBLP:conf/osdi/DeanG04} and Hadoop~\cite{hadoop},
  and can be found in most big data processing systems, like
  PigLatin~\cite{DBLP:conf/sigmod/OlstonRSKT08},
  Hive~\cite{TSJSCALWM09}, Dremmel~\cite{DBLP:journals/pvldb/MelnikGLRSTV10}.

Unlike traditional query processing, the complexity is no longer
dominated by the number of disk accesses.  Typically, a query is
evaluated by a sufficiently large number of servers such that the
entire data can be kept in the main memory of these servers.  The new
complexity bottleneck is the communication.  Typical network speeds in
large clusters are 1Gb/s, which is significantly lower than main
memory access.  In addition, any data reshuffling requires a global
synchronization of all servers, which also comes at significant cost;
for example, everyone needs to wait for the slowest server, and, worse,
in the case of a straggler, or a local node failure, everyone must
wait for the full recovery.  Thus, the dominating complexity
parameters in big data query processing are the number of
communication steps, and the amount of data being exchanged.

\paragraph{MapReduce-related models}

Several computation models have been proposed in order to understand
the power of MapReduce and related massively parallel programming
methods~\cite{DBLP:journals/talg/FeldmanMSSS10,DBLP:conf/soda/KarloffSV10,DBLP:conf/pods/KoutrisS11,DBLP:journals/corr/abs-1206-4377}.
These all identify the number of communication steps/rounds as a
main complexity parameter, but differ in their treatment of the
communication.

The first of these models 
was the MUD (Massive, Unordered,
Distributed) model of Feldman et
al.~\cite{DBLP:journals/talg/FeldmanMSSS10}.  It takes as input a
sequence of elements and applies a binary merge operation repeatedly,
until obtaining a final result, similarly to a User Defined Aggregate
in database systems.  The paper compares MUD with streaming
algorithms: a streaming algorithm can trivially simulate MUD, and the
converse is also possible if the merge operators are computationally
powerful (beyond PTIME).

Karloff et al.~\cite{DBLP:conf/soda/KarloffSV10} define
$\mathcal{MRC}$, a class of multi-round algorithms based on
using the MapReduce primitive as the sole building block, and fixing
specific parameters for balanced processing.  The
number of processors $p$ is $\Theta(N^{1-\epsilon})$, and each can
exchange MapReduce outputs expressible in $\Theta(N^{1-\epsilon})$
bits per step, resulting in $\Theta(N^{2-2\epsilon})$ total storage
among the processors on a problem of size $N$.  Their focus was
algorithmic, showing simulations of other parallel models by $\mathcal{MRC}$,
as well as the power of two round algorithms for specific problems.

Lower bounds for the single round MapReduce model are first discussed
by Afrati et al.~\cite{DBLP:journals/corr/abs-1206-4377}, who derive
an interesting tradeoff between reducer size and replication rate.
This is nicely illustrated by Ullman's drug interaction
example~\cite{DBLP:journals/crossroads/Ullman12}.  There are $n$
($=6,500$) drugs, each consisting of about 1MB of data about patients
who took that drug, and one has to find all drug interactions, by
applying a user defined function (UDF) to all pairs of drugs.  To see
the tradeoffs, it helps to simplify the example, by assuming we are
given {\em two} sets, each of size $n$, and we have to apply a UDF to
every pair of items, one from each set, in effect computing their
cartesian product.  There are two extreme ways to solve this. One can
use $n^2$ reducers, one for each pair of items; while each reducer has
size $2$, this approach is impractical because the entire data is
replicated $n$ times.  At the other extreme one can use a single
reducer that handles the entire data; the replication rate is 1, but
the size of the reducer is $2n$, which is also impractical.  As a
tradeoff, partition each set into $g$ groups of size $n/g$, and use
one reducer for each of the $g^2$ pairs of groups: the size of a
reducer is $2n/g$, while the replication rate is $g$.  Thus, there is
a tradeoff between the replication rate and the reducer size, which
was also shown to hold for several other classes of
problems~\cite{DBLP:journals/corr/abs-1206-4377}.

\paragraph{Towards lower bound models}

There are two significant limitations of this prior work: 
(1) As powerful and as convenient as the MapReduce framework is, the operations
it provides may not be able to take full advantage of the resource
constraints of modern systems.   The lower bounds say nothing about
alternative ways of structuring the computation that send and receive the same
amount data per step.  
(2) Even within the MapReduce framework, the only lower bounds apply to a 
single communication round, and say nothing about the limitations of multi-round
MapReduce algorithms.

While it is convenient that MapReduce hides the number of servers
from the programmer, when considering the most efficient way to use resources
to solve problems it is natural to expose information about those resources
to the programmer.
In this paper, we take the view that the number of servers $p$
should be an explicit parameter of the model, which allows us to focus
on the tradeoff between the amount of communication and the number of rounds.  
For example, going back to our cartesian product problem, if the
number of servers $p$ is known, there is one optimal way to solve the
problem: partition each of the two sets into $g = \sqrt{p}$ groups,
and let each server handle one pair of groups.

A model with $p$ as explicit parameter was proposed by Koutris and
Suciu~\cite{DBLP:conf/pods/KoutrisS11}, who showed both lower and
upper bounds for one round of communication.  In this model
only tuples are sent and they must be routed independent
of each other.  For example, \cite{DBLP:conf/pods/KoutrisS11} proves
that multi-joins on the same attribute can be computed in one round,
while multi-joins on different attributes, like $R(x),S(x,y),T(y)$
require strictly more than one round. The study was mostly focused
on understanding data skew, the model was limited, and the results do
not apply to more than one round.  

In this paper we develop more general
models, establish lower bounds that hold even in the absence of skew,
and use a bit model, rather than a tuple model, to represent data.

\paragraph{Our lower bound models and results}

We define the {\em Massively Parallel Communication} (\mpc) model,
to analyze the tradeoff between the number of rounds and the amount of
communication required in a massively parallel computing
environment.  We include the number of servers $p$ as a
parameter, and allow each server to be infinitely powerful, subject
only to the data to which it has access.
The model requires that each server receives only $O(N/p^{1-\varepsilon})$ 
bits of data at any step, where $N$ is the problem size, and $\varepsilon\in [0,1]$ is a
parameter of the model. This implies that the replication factor is
$O(p^\varepsilon)$ per round.
A particularly natural case is $\varepsilon=0$,
which corresponds to a replication factor of $O(1)$, or $O(N/p)$ bits per
server; $\varepsilon=1$ is degenerate, since it allows the entire data
to be sent to every server.

We establish both lower and upper bounds for computing a full
conjunctive query $q$, in two settings.  First, we restrict the
computation to a single communication round and examine the minimum
parameter $\varepsilon$ for which it is possible to compute $q$ with
$O(N/p^{1-\varepsilon})$ bits per processor; we call this the {\em space
exponent}. 
We show that the space exponent for connected queries is always at least
$1 -1/\tau^*(q)$, where $\tau^*(q)$ is the
{\em fractional (vertex) covering number} of the hypergraph associated with
$q$~\cite{DBLP:journals/siamdm/ChungFGG88}, which is the optimal value of the
vertex cover linear program (LP) for that hypergraph.  
This lower bound applies to the strongest possible model in which servers can
encode any information in their messages, and have access
to a common source of randomness.  This is stronger than the
lower bounds
in~\cite{DBLP:journals/corr/abs-1206-4377,DBLP:conf/pods/KoutrisS11},
which assume that the units being exchanged are tuples.

Our one round lower bound holds
even in the special case of {\em matching databases}, when all
attributes are from the same domain $[n]$ and
all input relations are (hypergraph) matchings, in other words, every relation has exactly $n$ tuples, and every attribute contains every value $1,2,\ldots,n$ exactly
once. Thus, the lower bound holds even in a case in which there is no
data skew. We describe a simple tuple-independent algorithm that  is easily
implementable in the MapReduce framework, which, in the special case
of matching databases,  matches our lower bound for any conjunctive query.
The algorithm uses the optimal solution for the fractional vertex
cover to find an optimal split of the input data to the servers. For
example, the linear query $L_2 = S_1(x,y),S_2(y,z)$ has an optimal
vertex cover $0,1,0$ (for the variables $x,y,z$), hence its space 
exponent is $\varepsilon=0$, whereas the cycle query 
$C_3 = S_1(x,y), S_2(y,z), S_3(z,x)$ has optimal vertex cover
$1/2,1/2,1/2$ and space exponent $\varepsilon = 1/3$.
We note that recent 
work~\cite{DBLP:conf/soda/GroheM06,DBLP:conf/focs/AtseriasGM08,DBLP:conf/pods/NgoPRR12}
gives  upper bounds on the query size in terms of a fractional {\em
  edge cover}, while our results are in terms of the {\em vertex cover}.
Thus, our first result is:

\begin{theorem} \label{th:intro1}
\label{th:one:round}
For every connected conjunctive query $q$,
any $p$-processor randomized \mpc\ algorithm computing $q$ in one round
requires space exponent $\varepsilon\ge 1-1/\tau^*(q)$.
This lower bound holds even over matching databases, for which it
is optimal.
\end{theorem}

Second, we establish lower bounds for multiple communication steps,
for a restricted version of the \mpc\ model, called \emph{tuple-based}
\mpc\ model.  The messages sent in the first round are still
unrestricted, but in subsequent rounds the servers can send only
tuples, either base tuples in the input tables, or join tuples
corresponding to a subquery; moreover, the destinations of each tuple
may depend only on the tuple content, the message received in the
first round, the server, and the round.  We note that any multi-step
MapReduce program is tuple-based, because in any map function the key
of the intermediate value depends only on the input tuple to the map
function.
Here, we prove that the number of rounds required is, essentially,
given by the depth of a query plan for the query, where each operator
is a subquery that can be computed in one round for the given
$\varepsilon$.  For example, to compute a length $k$ chain query
$L_k$, if $\varepsilon=0$, the optimal computation is a bushy join
tree, where each operator is $L_2$ (a two-way join) and the optimal
number of rounds is $\log_2 k$. If $\varepsilon=1/2$, then we can use
$L_4$ as operator (a four-way join), and the optimal number of rounds
is $\log_4 k$.  More generally, we can show nearly matching upper and
lower bounds based on graph-theoretic properties of the query such as
the following:

\begin{theorem}
 \label{th:intro2}
For space exponent $\varepsilon$, the number of rounds required for
any tuple-based \mpc\ algorithm to
compute any tree-like conjunctive query $q$ is at least
$\lceil\log_{k_\varepsilon} (\text{diam}(q))\rceil$
where $k_{\varepsilon} = 2 \lfloor  1/(1-\varepsilon) \rfloor$ and
$\text{diam}(q)$ is the diameter of $q$.  
Moreover, for any connected conjunctive query $q$, 
this lower bound is nearly matched (up to a difference of
essentially one round) by a tuple-based \mpc\ algorithm with space exponent
$\varepsilon$.
\end{theorem}

We further show that our results for conjunctive path queries imply that
any tuple-based \mpc\ algorithm with
space exponent $\varepsilon<1$ requires $\Omega(\log p)$ rounds to
compute the transitive closure or connected components of sparse
undirected graphs.
This is an interesting contrast to the results
of~\cite{DBLP:conf/soda/KarloffSV10}, which show that connected components
(and indeed minimum spanning trees) of undirected graphs can be computed in
only two rounds of
MapReduce provided that the input graph is sufficiently dense.

These are the first lower bounds that apply to multiple rounds of MapReduce.
Both lower bounds in \autoref{th:intro1} and \autoref{th:intro2} are
stated in a strong form: we show that any algorithm on the \mpc\ model
retrieves only a $1/p^{\Omega(1)}$ fraction of the answers to the
query in expectation, when the inputs are drawn uniformly at random
(the exponent depends on the query and on $\varepsilon$); Yao's
Lemma~\cite{yao83} immediately implies a lower bound for any
randomized algorithm over worst-case inputs.  Notice that the fraction
of answers gets worse as the number of servers $p$ increases.  In
other words, the more parallelism we want, the worse an algorithm
performs, if the number of communication rounds is bounded.

\paragraph{Related work in communication complexity}

The results we show belong to the study of communication complexity,
for which there is a very large body of existing research~\cite{kn97}.
Communication complexity considers the number of bits that need to be
communicated between cooperating agents in order to solve
computational problems when the agents have unlimited computational
power.  Our model is related to the so-called number-in-hand
multi-party communication complexity, in which there are multiple
agents and no shared information at the start of communication.  This
has already been shown to be important to understanding the processing
of massive data: Analysis of number-in-hand (NIH) communication
complexity has been the main method for obtaining lower bounds on the
space required for data stream algorithms~(e.g. \cite{ams:freq}).

However, there is something very different about the results that we prove here.
In almost all prior lower bounds, there is at least one agent that has
access to all communication between agents\footnote{\footnotesize Though private-messages
models have been defined before, we are aware of only two lines of work where
lower bounds make use of the fact that no single agent has access to all
communication:  
(1) Results of \cite{DBLP:conf/focs/GalG07,DBLP:conf/icalp/GuhaH09}
use the assumption that communication is both private and (multi-pass) one-way,
but unlike the bounds we prove here, their lower bounds 
are smaller than the total input size;
(2) Tiwari~\cite{tiw87} defined a distributed model of communication 
complexity in networks in which in input is given to two processors that
communicate privately using other helper processors.
However, this model is equivalent to ordinary public two-party
communication when the network allows direct private communication between any
two processors, as our model does.}.
(Typically, this is either
via a shared blackboard to which all agents have access or a referee who
receives all communication.)   
In this case, no problem on $N$ bits whose answer is
$M$ bits long can be shown to require more than $N+M$ bits of communication.

In our \mpc\ model, all communication between servers is {\em private} and
we restrict the communication per processor per step, rather than the total
communication.  
Indeed, the privacy of communication is essential to our lower bounds, since
we prove lower bounds that apply when the total communication is much larger
than $N+M$. (Our lower bounds for some problems apply when the total
communication is as large as $N^{1+\delta}$.)

%% file: model.tex
\section{Preliminaries}
\label{sec:defs}

\subsection{Massively Parallel Communication}
\label{subsec:model}

We fix a parameter $\varepsilon \in [0,1]$, called the {\em space
  exponent}, and define the \mpc($\varepsilon$) model as follows.  The
computation is performed by $p$ servers, called {\em workers}, 
connected by a complete network of private channels. 
The input data has size $N$ bits, and is
initially distributed evenly among the $p$ workers.  The computation
proceeds in rounds, where each round consists of local computation at
the workers interleaved with global communication.  The complexity is
measured in the number of communication rounds.  The servers have
unlimited computational power, but there is one important restriction:
at each round, a worker may receive a total of only
$O(N/p^{1-\varepsilon})$ bits of data from all other workers
combined.  Our goal is to find lower and upper bounds on the number of
communication rounds.

The space exponent represents the degree of replication during
communication; in each round, the total amount of data exchanged is
$O(p^\varepsilon)$ times the size of the input data.  When
$\varepsilon=0$, there is no replication, and we call this the
basic \mpc\ model. The case $\varepsilon=1$ is degenerate because each
server can receive the entire data, and any problem can be solved in
a single round.  Similarly, for any fixed $\varepsilon$, if we allow
the computation to run for $\Theta(p^{1-\varepsilon})$ rounds, the
entire data can be sent to every server and the model is again
degenerate.

We denote $M_{uv}^r$ the message sent by server $u$ to server $v$
during round $r$ and denote $M_v^r=(M_v^{r-1},(M_{1v}^r,\ldots,M_{pv}^r))$ 
the concatenation of all messages sent to $v$ up to round $r$.  Assuming $O(1)$ 
rounds, each message $M_v^r$ holds $O(N/p^{1-\varepsilon})$ bits.  For our
multi-round lower bounds in Section~\ref{sec:multistep}, we will
further restrict what the workers can encode in the messages
$M^r_{uv}$ during rounds $r \geq 2$.

\subsection{Randomization}

The \mpc\ model allows randomization.  The random bits are
available to all servers, and are computed independently of the input
data.  The algorithm may fail to produce its output with a small
probability $\eta>0$, independent of the input.  For example, we use
randomization for load balancing, and abort the computation if the
amount of data received during a communication would exceed the
$O(N/p^{1-\varepsilon})$ limit, but this will only happen with
exponentially small probability.

To prove lower bounds for randomized algorithms, we use Yao's
Lemma~\cite{yao83}.  We first prove bounds for {\em deterministic}
algorithms, showing that any algorithm fails with probability at
least $\eta$ over inputs chosen randomly from a distribution $\mu$. 
This implies, by Yao's Lemma, that every randomized algorithm with the 
same resource bounds will fail on some input (in the support of $\mu$) with 
probability at least $\eta$ over the algorithm's random choices.

\subsection{Conjunctive Queries}
\label{subsec:cq}

In this paper we consider a particular class of problems for the \mpc\
model, namely computing answers to conjunctive queries over an input
database.  We fix an input vocabulary $S_1, \ldots, S_\ell$, where
each relation $S_j$ has a fixed arity $a_j$; we denote $a = \sum_{j
  =1}^{\ell} a_j$.  The input data consists of one relation instance
for each symbol.  We denote $n$ the largest number of tuples in any
relation $S_j$; then, the entire database instance can be encoded using
$N = O(n \log n)$ bits, because $\ell = O(1)$ and $a_j = O(1)$ for
$j=1,\dots, \ell$.

We consider full conjunctive queries (CQs) without
self-joins, denoted as follows:
\begin{equation} \label{eq:q}
  q(x_1,\ldots, x_k) = S_1(\bar x_1), \ldots, S_\ell(\bar x_\ell) 
\end{equation}

The query is {\em full}, meaning that every variable in the body
appears the head (for example $q(x) = S(x,y)$ is not full), and {\em
  without self-joins}, meaning that each relation name $S_j$ appears
only once (for example $q(x,y,z) = S(x,y), S(y,z)$ has a
self-join). The {\em hypergraph} of a query $q$ is defined by
introducing one node for each variable in the body and one hyperedge
for each set of variables that occur in a single atom. We say that a
conjunctive query is {\em connected} if the query hypergraph is
connected (for example, $q(x,y) = R(x),S(y)$ is not connected).
We use $\text{vars}(S_j)$ to denote the set of variables in the atom $S_j$, and
$\text{atoms}(x_i)$ to denote the set of atoms where $x_i$ occurs; $k$
and $\ell$ denote the number of variables and atoms in $q$, as
in~\eqref{eq:q}. The {\em connected components} of $q$ are the maximal 
connected subqueries of $q$. \autoref{tab:queries} illustrates example queries 
used throughout this paper.

We consider two query evaluation problems. In {\sc Join-Reporting}, we
require that all tuples in the relation defined by $q$ be produced. In
{\sc Join-Witness}, we require the production of at least one tuple
in the relation defined by $q$, if one exists;
{\sc Join-Witness} is the verified version of the natural decision problem
{\sc Join-NonEmptiness}.

\paragraph{Characteristic of a Query} 

The {\em characteristic} of a conjunctive query $q$ as in \eqref{eq:q} is
defined as $\chi(q) = k +\ell - \sum_j a_j - c$, where $k$ is the
number of variables, $\ell$ is the number of atoms, $a_j$ is the
arity of atom $S_j$, and $c$ is the number of connected components of $q$.

For a query $q$ and a set of atoms $M \subseteq \text{atoms}(q)$, define
$q/M$ to be the query that results from contracting the edges in the hypergraph
of $q$. As an example, for the query $L_5$ in \autoref{tab:queries}, 
$L_5/ \{S_2, S_4 \}= S_1(x_0, x_1), S_3(x_1,x_3), S_5(x_3,x_5)$.

\begin{lemma} \label{lemma:chi} The characteristic of a query $q$ satisfies the following 
properties:
\begin{itemize}
\item[(a)] If $q_1, \ldots, q_c$ are the connected components of $q$, then
$\chi(q)=\sum_{i=1}^c \chi(q_i)$.
\item[(b)] For any $M \subseteq \text{atoms}(q)$, $\chi(q/M)=\chi(q) -\chi(M)$.
\item[(c)] $\chi(q) \leq 0$.
\item[(d)] For any $M \subseteq \text{atoms}(q)$, $\chi(q) \leq \chi(q/M)$.
\end{itemize}
\end{lemma}

\begin{proof}
Property (a) is immediate from the definition of $\chi$, since the connected
components of $q$ are disjoint with respect to variables and atoms.
Since $q/M$ can be produced by contracting according to each connected
component of $M$ in turn, by property (a) and induction it suffices to show
that property (b) holds in the case that $M$ is connected.  
If a connected $M$ has $k_M$ variables,
$\ell_M$ atoms, and total arity $a_M$, then the query after contraction, $q/M$,
will have the same number of connected components, $k_M-1$ fewer variables,
and the terms for the number of atoms and total arity will be reduced by
$\ell_M-a_M$ for a total reduction of $k_M+\ell_M-a_M-1=\chi(M)$.  
Thus, property (b) follows.

By property (a), it suffices to prove (c) when $q$ is connected. 
If $q$ is a single atom then $\chi(q) \leq 0$, since the number
of variables is at most the arity of the atom in $q$.
We reduce to this case by repeatedly contracting the atoms of $q$ until 
only one remains and showing that $\chi(q)\le \chi(q/S_j)$:  
Let $m\le a_j$ be the number of distinct variables in atom $S_j$.
Then, $\chi(q/S_j) = (\ell-1) + (k-m+1) - (a - a_j) - 1 = \chi(q)
+(a_j-m) \geq \chi(q)$.
Property (d) also follows inductively from $\chi(q)\le \chi(q/S_j)$ or 
by the combination of property (b) and property (c) applied to $M$.
\end{proof}

\begin{sloppypar}
Finally, let us call a query $q$ {\em tree-like} if $q$ is connected and
$\chi(q) = 0$.  For example, the query 
$L_k$ is tree-like, and so is any query over a binary 
vocabulary whose graph is a tree. Over non-binary vocabularies, any tree-like 
query is acyclic, but the converse does not hold: $q =
S_1(x_0,x_1,x_2),S_2(x_1,x_2,x_3)$ is acyclic but not tree-like. 
An important property of tree-like queries is that every connected subquery will
be also tree-like.
\end{sloppypar}

\paragraph{Vertex Cover and Edge Packing}

A {\em fractional vertex cover} of a query $q$ is any feasible
solution of the LP shown on the left of \autoref{tab:LP}.  The vertex
cover associates a non-negative number $u_i$ to each variable $x_i$
s.t. every atom $S_j$ is ``covered'', $\sum_{i: x_i \in
  \text{vars}(S_j)} v_i \geq 1$.
The dual LP corresponds to a {\em fractional edge packing} problem (also known
as a {\em fractional matching} problem),
which associates non-negative numbers $u_j$ to each atom $S_j$.  The two
LPs have the same optimal value of the objective function, known as the
{\em fractional covering number}~\cite{DBLP:journals/siamdm/ChungFGG88} of the hypergraph associated with $q$
and denoted by $\tau^*(q)$.  Thus, $\tau^*(q) = \min \sum_i v_i = \max
\sum_j u_j$.  Additionally, if all inequalities are satisfied as
equalities by a solution to the LP, we say that the solution is {\em
  tight}.

\begin{figure}[b]
\centering{
\begin{tabular}[c]{|l|l|} \hline
  Vertex Covering LP & Edge Packing LP \\ \hline
\parbox{5cm}{
   \begin{align}
    & \forall j \in [\ell]: 
     \sum_{i: x_i \in \text{vars}(S_j)} \kern -1em v_i \geq 1 \label{eq:cover} \\
    & \forall i \in [k] : v_i \geq 0 \nonumber
  \end{align}
}
&
\parbox{5cm}{
  \begin{align}
    & \forall i \in [k]: 
     \sum_{j: x_i \in \text{vars}(S_j)} \kern -1em u_j \leq 1 \label{eq:cover:dual} \\
    &  \forall j \in [\ell] : u_j \geq 0 \nonumber
  \end{align}
}
\\ \hline
$\text{minimize } \sum_{i=1}^k v_i$ &
$\text{maximize } \sum_{j=1}^{\ell} u_j$ \\ \hline
\end{tabular}
}
\caption{The vertex covering LP of the hypergraph of a query $q$, and its
dual edge packing LP.}
\label{tab:LP}
\end{figure}

\begin{sloppypar}
\begin{example}
For a simple example, a fractional vertex cover of the
query\footnote{We drop the head variables when clear from the
  context.} $L_3 = S_1(x_1, x_2),S_2(x_2,x_3), S_3(x_3, x_4)$ is any
solution to $v_1+v_2 \geq 1$, $v_2+v_3 \geq 1$ and $v_3+v_4 \geq 1$;
the optimal is achieved by $(v_1,v_2,v_3,v_4)=(0,1,1,0)$, which is not
tight.  An edge packing is a solution to $u_1 \leq 1$, $u_1+u_2 \leq
1$, $u_2 + u_3 \leq 1$ and $u_3 \leq 1$, and the optimal is achieved
by $(1,0,1)$, which is tight.
\end{example}
\end{sloppypar}

The fractional edge {\em packing} should not be confused with the
fractional edge {\em cover}, which has been used recently in several
papers to prove bounds on query size and the running time of a
sequential algorithm for the
query~\cite{DBLP:conf/focs/AtseriasGM08,DBLP:conf/pods/NgoPRR12}; for
the results in this paper we need the fractional packing.  The two
notions coincide, however, when they are tight.

\subsection{Input Servers}

We assume that, at the beginning of the algorithm, each relation $S_j$
is stored on a separate server, called an {\em input server}, which
during the first round sends a message $M^1_{ju}$ to every worker $u$.
After the first round, the input servers are no longer used in the
computation.  All lower bounds in this paper assume that the relations
$S_j$ are given on separate input servers.  All upper bounds hold for
either model.

The lower bounds for the model with separate input servers carry over
immediately to the standard \mpc\ model, because any algorithm in the
standard model can be simulated in the model with separate input
servers.  Indeed, the algorithm must compute the output correctly for
any initial distribution of the input data on the $p$ servers: we
simply choose to distribute the input relations $S_1, \ldots, S_\ell$
such that the first $p/\ell$ servers receive $S_1$, the next $p/\ell$
servers receive $S_2$, etc., then simulate the algorithm in the model
with separate input servers (see~\cite[proof of Proposition
3.5]{DBLP:conf/pods/KoutrisS11} for a detailed discussion).  Thus, it
suffices to prove our lower bounds assuming that each input relation
is stored on a separate input server.  In fact, this model is even more
powerful, because an input server has now access to the entire
relation $S_j$, and can therefore perform some global computation on
$S_j$, for example compute statistics, find outliers, etc., which are
common in practice.

\subsection{Input Distribution}

We find it useful to consider input databases of the following form that
we call a {\em matching} database:
The domain of the input database will be $[n]$, for $n>0$.  
In such a database each relation $S_j$ is an {\em $a_j$-dimensional
  matching}, where $a_j$ is its arity.  In other words, $S_j$ has
exactly $n$ tuples and each of its columns contains exactly the values
$1, 2, \ldots, n$; each attribute of $S_j$ is a key.  For example, if
$S_j$ is binary, then an instance of $S_j$ is a permutation on $[n]$;
if $S_j$ is ternary then an instance consists of $n$ node-disjoint
triangles.  Moreover, the answer to a connected conjunctive query $q$ on a
matching database is a table
where each attribute is a key, because we have assumed that $q$ is
full; in particular, the output to $q$ has at most $n$
tuples.  In our lower bounds we assume that a matching database is
randomly chosen with uniform probability, for a fixed $n$.

Matching databases are database instances {\em without skew}.  By
stating our lower bounds on matching databases we make them even stronger,
because they imply that a query cannot be computed even in the absence
of skew; of course, the lower bounds also hold for arbitrary
instances.  Our upper bounds, however, hold only on matching databases.  Data
skew is a known problem in parallel processing, and requires
dedicated techniques. Lower and upper bounds
accounting for the presence of skew are discussed
in~\cite{DBLP:conf/pods/KoutrisS11}.

\begin{table*}
  \centering
  \renewcommand{\arraystretch}{1.8}
\resizebox{\textwidth}{!}
{
  \begin{tabular}[c]{|l|c|c|c|c|c|} \hline
    Conjunctive Query & Expected  & Minimum      & Share         &   Value & Space \\
          & answer size       & Vertex Cover &  Exponents&   $\tau^*(q)$ & Exponent \\ \hline
$C_k(x_1,\ldots,x_k) = \bigwedge_{j=1}^{k} S_j(x_j,x_{(j \bmod k)+1}) $ & $1$ & $\frac{1}{2}, \dots, \frac{1}{2}$ 
 & $\frac{1}{k}, \dots, \frac{1}{k}$& $k/2$ & $1- 2/k$ \\ \hline
$T_k(z, x_1, \ldots, x_k)= \bigwedge_{j=1}^k S_j(z,x_j) $ & $n$ & $1,0,\dots,0$ & $1, 0,\dots, 0$ & $1$ & $0$ \\ \hline
$L_k(x_0, x_1,\ldots, x_k) =  \bigwedge_{j=1}^k S_j(x_{j-1},x_j)$  & $n$ & $0,1,0,1,\dots$
  & $0, \frac{1}{\lceil k/2 \rceil},0,\frac{1}{\lceil k/2 \rceil}, \dots$& $\lceil k/2 \rceil$ & $1 - 1/\lceil k/2 \rceil$\\ \hline
$B_{k,m}(x_1, \ldots, x_k) = \bigwedge_{I \subseteq [k], |I|=m} S_{I}(\bar x_{I})$  & $n^{k-(m-1)\binom{k}{m}}$ 
 & $\frac{1}{m}, \dots, \frac{1}{m}$ & $\frac{1}{k}, \dots, \frac{1}{k}$ & $k/m$ & $1- m/k$\\ \hline
  \end{tabular}
}
  \caption{Running examples in this paper: $C_k=$ cycle query, $L_k=$ linear
    query, $T_k=$ star query, and $B_{k,m}=$ query with $\binom{k}{m}$ relations, where each relation
    contains a distinct set of $m$ out of the $k$ head variables.  Assuming the
    inputs are random permutation, the answer sizes represent exact values
    for $L_k, T_k$, and expected values for $C_k, B_{k,m}$.  
   }
  \label{tab:queries}
\end{table*}

\ifnotentropy{}{
\subsection{Entropy}

\paris{I moved the entropy discussion here.}
 
Let us fix a finite probability space. For random variables $X$
and $Y$, the {\em entropy} and the {\em conditional entropy} are 
defined as follows:
\begin{align}
  H(X) = & - \sum_x P(X=x)  \log P(X=x) \nonumber \\
  H(X \mid Y) = & \sum_y P(Y=y) H(X \mid Y=y)  \label{eq:cond:ent:v}
\end{align}
The entropy satisfies the following basic inequalities, assuming $X$
has a support of size $n$:
\begin{align}
  H(X) &\leq  \log n  \nonumber \\
  H(X \mid Y) &\leq  H(X) \nonumber \\
  H(X, Y) &=  H(X \mid Y) + H(Y) \label{eq:cond:ent}
\end{align}
}

\subsection{Friedgut's Inequality}

Friedgut~\cite{friedgut2004hypergraphs} introduces
the following class of inequalities.  Each inequality is described by
a hypergraph, which in our paper corresponds to a query, so we will
describe the inequality using query terminology.  Fix a query $q$ as
in \eqref{eq:q}, and let $n > 0$.  For every atom $S_j(\bar x_j)$ of
arity $a_j$, we introduce a set of $n^{a_j}$ variables $w_{j}(\ba_j) \geq  0$, 
where $\ba_j \in [n]^{a_j}$. If $\ba \in [n]^a$, we denote by $\ba_j$ the
vector of size $a_j$ that results from projecting on the variables of the
relation $S_j$.
Let $\mathbf{u} = (u_1, \dots, u_{\ell})$ be a 
fractional {\em edge cover} for $q$.  Then:
\begin{align}
  \sum_{\ba \in [n]^k} \prod_{j=1}^{\ell} w_{j}( \ba_j) \leq & \prod_{j=1}^{\ell}
  \left(\sum_{\ba_j \in [n]^{a_j}}  w_{j} (\ba_j)^{1/u_j}\right)^{u_j} \label{eq:friedgut}
\end{align}
We illustrate Friedgut's inequality on $C_3$ and $L_3$:
\begin{align}
C_3(x,y,z) = S_1(x,y),S_2(y,z),S_3(z,x) \nonumber\\
L_3(x,y,z,w) = S_1(x,y),S_2(y,z),S_3(z,w) \label{eq:ccc}
\end{align}
$C_3$ has cover $(1/2,1/2,1/2)$, and $L_3$ has cover $(1,0,1)$.  Thus, we
obtain the following inequalities, where $\alpha,\beta,\gamma$ stand for
$w_1,w_2,w_3$ respectively:
\begin{align*}
  \sum_{x,y,z \in [n]}\kern -1em \alpha_{xy}\cdot \beta_{yz} \cdot \gamma_{zx} \leq &
  \sqrt{\sum_{x,y \in [n]} \alpha_{xy}^2 \sum_{y,z \in [n]} \beta_{yz}^2
    \sum_{z,x \in [n]} \gamma_{zx}^2} \\
 \kern -2em \sum_{\kern +2em x,y,z,w \in [n]}\kern -2.2em  \alpha_{xy}\cdot \beta_{yz} \cdot \gamma_{zw} \leq &
  \sum_{x,y \in [n]} \alpha_{xy}\,\cdot\, \max_{y,z \in [n]}
    \beta_{yz}\,\cdot\,\sum_{z,w \in [n]} \gamma_{zw}
\end{align*}
where we used the fact that $\lim_{u\rightarrow 0} (\sum
\beta_{yz}^{\frac{1}{u}})^u = \max \beta_{yz}$.

Friedgut's inequalities immediately imply a well known result
developed in a series of
papers~\cite{DBLP:conf/soda/GroheM06,DBLP:conf/focs/AtseriasGM08,DBLP:conf/pods/NgoPRR12}
that gives an upper bound on the size of a query answer as a function
on the cardinality of the relations.  For example in the case of
$C_3$, consider an instance $S_1, S_2, S_3$, and set $\alpha_{xy} = 1$ if
$(x,y) \in S_1$, otherwise $\alpha_{xy}=0$ (and similarly for
$\beta_{yz},\gamma_{zx}$).  We obtain then $ |C_3| \leq \sqrt{|S_1| \cdot |S_2|
  \cdot |S_3|}$.  Note that all these results are expressed in terms
of a fractional edge {\em cover}.  When we apply Friedgut's inequality
in Section~\ref{subsec:onestep} to a fractional edge {\em packing}, we 
ensure that the packing is tight.

%% file: onestep.tex
\section{One Communication Step}
\label{sec:onestep}

Let the {\em space exponent} of a query $q$ be the smallest $\varepsilon
\geq 0$ for which $q$ can be computed using one communication step in
the \mpc($\varepsilon$) model.  In this section, we prove
Theorem~\ref{th:one:round}, which gives both a general lower bound
on the space exponent for evaluating connected conjunctive queries
and a precise characterization of the space exponent for evaluating them
them over matching databases.
The proof consists of two parts: we show the optimal algorithm in
\ref{subsec:algorithms}, and then present the matching lower bound
in \ref{subsec:onestep}. We will assume w.l.o.g. throughout this
section that the queries do not contain any unary relations.   Indeed,
by definition,  the only unary matching relation is the set
$\{1,2,\dots,n\}$, and hence it is trivially known to all servers, so we
can simply remove all unary relations before evaluating the query.

\subsection{An Algorithm for One Round}
\label{subsec:algorithms}

We describe here an algorithm, which we call \textsc{HyperCube} (HC),
 that computes a conjunctive query in one
step.  It uses ideas that can be traced back to
Ganguly~\cite{DBLP:journals/jlp/GangulyST92} for parallel processing
of Datalog programs, and were also used by Afrati and
Ullman~\cite{DBLP:conf/edbt/AfratiU10} to optimize joins in
MapReduce, and by Suri and
Vassilvitskii~\cite{DBLP:conf/www/SuriV11} to count triangles. 

Let $q$ be a query as in \eqref{eq:q}.  Associate to each variable $x_i$
a real value $e_i \geq 0$, called the {\em share exponent} of $x_i$,
such that $\sum_{i=1}^k e_i = 1$.  If $p$ is the number of servers,
define $p_i = p^{e_i}$: these values are called {\em
  shares}~\cite{DBLP:conf/edbt/AfratiU10}.  We assume that the shares
are integers.  Thus, $p = \prod_{i=1}^k p_i$, and each server can be
uniquely identified with a point in the $k$-dimensional hypercube $[p_1]
\times \dots \times [p_k]$.

The algorithm uses $k$ independently chosen random hash functions $h_i
: [n] \rightarrow [p_i]$, one for each variable $x_i$.  During the
communication step, the algorithm sends every tuple $S_j(\ba_j) =
S_j(\alpha_{i_1}, \dots, \alpha_{i_{a_j}})$ to all servers $\mathbf{y} \in [p_1]
\times \dots \times [p_k]$ such that $h_{i_m}(\alpha_{i_m}) =
\mathbf{y}_{i_m}$ for any $1 \leq m \leq a_j$.  In other words, the
tuple $S_j(\ba_j)$ knows the server number along the dimensions $i_1,
\ldots, i_{a_j}$, but does not know the server number along the other
dimensions, and there it needs to be replicated.  After receiving the data, 
each server outputs all
query answers derivable from the received data.  The algorithm finds
all answers, because each potential output tuple $(\alpha_1, \ldots, \alpha_k)$
is known by the server $\mathbf{y} = (h_1(\alpha_1), \dots, h_k(\alpha_k))$.

\begin{example}
  We illustrate how to compute $C_3(x_1,x_2,x_3) = S_1(x_1,
  x_2),S_2(x_2, x_3),S_3(x_3, x_1)$.
  Consider the share exponents $e_1=e_2=e_3=1/3$. Each of the $p$ servers is
  uniquely identified by a triple $(y_1,y_2,y_3)$, where $y_1, y_2,
  y_3\in [p^{1/3}]$.  In the first communication round, the input
  server storing $S_1$ sends each tuple $S_1(\alpha_1,\alpha_2)$ to all servers
  with index $(h_1(\alpha_1), h_2(\alpha_2), y_3)$, for all $y_3 \in [p^{1/3}]$:
  notice that each tuple is replicated $p^{1/3}$ times.  The input
  servers holding $S_2$ and $S_3$ proceed similarly with their tuples.
  After round 1, any three tuples $S_1(\alpha_1, \alpha_2)$,
  $S_2(\alpha_2, \alpha_3)$, $S_3(\alpha_3, \alpha_1)$ that contribute to the output tuple
  $C_3(\alpha_1,\alpha_2,\alpha_3)$ will be seen by the server $\mathbf{y} = (h_1(\alpha_1),
  h_2(\alpha_2), h_3(\alpha_3))$: any server that detects three matching tuples
  outputs them.
\end{example}

\begin{proposition} \label{th:onestep:upper} Fix a fractional vertex
  cover $\mathbf{v} = (v_1, \ldots, v_k)$ for a connected conjunctive
  query $q$, and let $\tau = \sum_i v_i$.  The  HC algorithm with share
  exponents $e_i=v_i/\tau$ computes $q$ on any matching database 
  in one  round in $\mpc(\varepsilon)$, where $\varepsilon = 1-1/\tau$,
  with probability of failure $\eta \leq \text{exp}(-O(n/p^\varepsilon))$. 
\end{proposition}

This proves the optimality claim of \autoref{th:one:round}:
choose a vertex cover with value $\tau^*(q)$, the fractional covering
number of $q$.  Proposition~\ref{th:onestep:upper} shows that $q$ can
be computed in one round in $\mpc(\varepsilon)$, with
$\varepsilon = 1 -1/\tau^*$.

\begin{proof}
  Since $\mathbf{v}$ forms a fractional vertex cover, for every
  relation symbol $S_j$ we have $\sum_{i: x_i \in \text{vars}(S_j)}
  e_i \geq 1/\tau$. Therefore, $\sum_{i: x_i \not\in
    \text{vars}(S_j)} e_i \leq 1 - 1/\tau$.  Every tuple
  $S_j(\ba_j)$ is replicated $\prod_{i: x_i \not\in \text{vars}(S_j)}
  p_i \leq p^{1-1/\tau}$ times.  Thus, the total number of tuples
   that are received by all servers is $O(n \cdot
  p^{1-1/\tau})$.  We claim that these tuples are uniformly
  distributed among the $p$ servers: this proves the theorem, since
  then each server receives $O(n/p^{1/\tau})$ tuples.

  To prove the claim, we note that for each tuple $t \in S_j$, the
  probability over the random choices of the hash functions $h_1,
  \ldots, h_k$ that the tuple is sent to server $s$ is precisely
  $\prod_{i:x_i
    \in \text{vars}(S_j)} p_i^{-1}$. Thus, the expected number of
  tuples from $S_j$ sent to $s$ is $n/\prod_{i:x_i \in S_j} p_i \leq
  n/p^{1-\varepsilon}$.  Since $S_j$ is an $a_j$-matching, different
  tuples are sent by the random hash functions to independent
  destinations, since any two tuples differ in every attribute.  Using
  standard Chernoff bounds, we derive that the probability that the
  actual number of tuples per server deviates more than a constant
  factor from the expected number is $\eta \leq
  \text{exp}(-O(n/p^{1-\varepsilon}))$.
\end{proof}

\subsection{A Lower Bound for One Round}
\label{subsec:onestep}

For a fixed $n$, consider a probability distribution where the input
$I$ is chosen randomly, with uniform probability from all matching
database instances.  Let $\mathbf{E}[|q(I)|]$ denote the expected
number of answers to the query $q$.  We prove in this
section\footnote{Recall that we have assumed that $q$ has no unary
  relations.  Otherwise, the theorem fails, as illustrated by the
  query $q = S_1(x), S_2(x,y),S_3(y)$, which has $\tau^* = 2$, yet can
  be computed with space exponent $\varepsilon = 0$, on arbitrary
  databases (not only matching databases). Indeed, notice that both
  unary relations $S_1, S_3$ require $n$ bits to be represented (as bit
  vectors), whereas $S_2$ requires $\Omega(n \log n)$ bits. Hence,
  if $p \leq \log n$, $S_1, S_2$ can be broadcast to every server.}:

\begin{theorem} \label{th:onestep} Let $q$ be a connected conjunctive query,
 let $\tau^*$ be
  the fractional covering number of $q$, and $\varepsilon < 1 -
  1/\tau^*$.  Then, any deterministic \mpc($\varepsilon$) algorithm
  that runs in one communication round on $p$ servers reports
  $O(\mathbf{E}[|q(I)|]/p^{\tau^*(1-\varepsilon)-1})$ answers in
  expectation.
\end{theorem}

In particular, the theorem implies that the space exponent of $q$ is
at least $1-1/\tau^*$. Before we prove the theorem, we show how to
extend it to randomized algorithms using Yao's principle.  For this,
we show a lemma that we also need later.

\begin{lemma} 
\label{lem:expected_size} 
The expected number of answers
  to connected query $q$ is $\mathbf{E}[|q(I)|] = n^{1+\chi(q)}$, where the
  expectation is over a uniformly chosen matching database I.
\end{lemma}

\begin{proof}
  For any relation $S_j$, and any tuple $\mathbf{a}_j \in [n]^{a_j}$, 
  the probability that $S_j$ contains $\mathbf{a}_j$
  is $\P(\mathbf{a_j} \in S_j) = n^{1-a_j}$.  Given a tuple $\mathbf{a}
  \in [n]^k$ of the same arity as the query answer, let
  $\mathbf{a}_j$ denote its projection on the variables in $S_j$.  Then:
  \begin{align*}
    \mathbf{E}[|q(I)|]  
    & = \sum_{\ba \in [n]^k}  \P(\bigwedge_{j=1}^{\ell} (\mathbf{a}_j \in S_j)) \\
    & = \sum_{\ba \in [n]^k} \prod_{j=1}^{\ell} \P(\mathbf{a}_j \in S_j)  \\
    & = \sum_{\ba \in [n]^k} \prod_{j=1}^{\ell} n^{1-a_j} \\
    & = n^{k+\ell-a}
\end{align*}
Since query $q$ is connected, $k+\ell-a = 1+\chi(q)$ and hence
$\mathbf{E}[|q(I)|] = n^{1+\chi(q)}$.
\end{proof}

\autoref{th:onestep} and \autoref{lem:expected_size}, together with
Yao's lemma, imply the following lower bound for
randomized algorithms. 

\begin{corollary} \label{cor:yao}
Let $q$ be any connected conjunctive query.
Any one round randomized \mpc($\varepsilon$) algorithm 
with $p=\omega(1)$ and $\varepsilon < 1 - 1/\tau^*(q)$ 
fails to compute $q$ with probability
$\eta = \Omega(n^{\chi(q)}) = n^{-O(1)}$.
\end{corollary}

\begin{proof}
Choose a matching database $I$ input to $q$ uniformly at random.
Let $A(I)$ denote the set of correct answers returned by the
algorithm on $I$: $A(I)\subseteq q(I)$.
Observe that the algorithm fails on $I$ iff $|q(I)-A(I)|>0$.

Let $\gamma= 1/p^{\tau^*(q)(1-\varepsilon)-1}$.
Since $p = \omega(1)$ and
$\varepsilon < 1 - 1/\tau^*(q)$, it follows that $\gamma = o(1)$.
By \autoref{th:onestep}, for any deterministic 
one round \mpc($\varepsilon$) algorithm 
we have $\E[|A(I)|]=O(\gamma) \E[|q(I)|]$ and hence, by
\autoref{lem:expected_size},
$$\E[|q(I)-A(I)|]=(1-o(1)) \E[|q(I)|]=(1-o(1))n^{1+\chi(q)}$$
However, we also have that
$$\E[|q(I)-A(I)|]\le \P[|q(I)-A(I)|>0]\cdot \max_I |q(I)-A(I)|.$$

Since $|q(I) - A(I)|\le |q(I)| \leq n$ for all $I$, we see that the
failure probability of the algorithm for randomly chosen $I$,
$\P[|q(I)-A(I)|>0]$, is at least $\eta=(1-o(1))n^{\chi(q)}$ which is
$n^{-O(1)}$ for any $q$. 
Yao's lemma implies that every one round randomized \mpc($\varepsilon$)
algorithm will fail to compute
$q$ with probability at least $\eta$ on some matching database input.
\end{proof}

In the rest of the section we prove \autoref{th:onestep}, which deals
with one-round deterministic algorithms and random matching 
databases $I$. Let us fix some server and let $m(I)$ denote the function 
specifying the message the server receives on input $I$.
Intuitively, this server can only report those tuples that it knows
are in the input based on the value of $m(I)$. To make this notion precise, 
for any fixed value $m$ of $m(I)$, define the set of tuples of a relation $R$
of arity $a$ {\em known} by the server given message $m$ as
\begin{align*}
K_m(R)=\{t\in [n]^a\mid \mbox{ for all matching databases }I,
m(I)=m\Rightarrow t\in R(I) \}
\end{align*}

We will particularly apply this definition with $R=S_j$ and $R=q$.
Clearly, an output tuple $\ba \in K_m(q)$ iff for every $j$,
$\ba_j \in K_m(S_j)$, where $\ba_j$ denotes the projection of $\ba$ on
the variables in the atom $S_j$.

We will first prove an upper bound for each $|K_m(S_j)|$
in Section~\ref{subsec:relation_bound}.
Then in Section~\ref{subsec:onestepB} we use this bound,
along with Friedgut's inequality, to establish an upper bound for
$|K_m(q)|$ and hence prove \autoref{th:onestep}.

\subsubsection{Bounding the Knowledge of Each Relation}
\label{subsec:relation_bound}

Let us fix a server, and an input relation $S_j$. Observe that, for a randomly
chosen matching database $I$, $S_j$ is a uniformly chosen $a_j$-dimensional
matching. There are precisely $(n!)^{a_j-1}$ different $a_j$-dimensional matchings
of arity $a_j$ and thus the number of bits $N$ necessary to represent the relation
is $(a_j-1)\log(n!)$. 

Let $m(S_j)$ be the part of the message $m$ received from the server that corresponds to 
$S_j$. The following lemma provides a bound on the expected knowledge $K_{m(S_j)}(S_j)$ 
the server may obtain from $S_j$: 

\begin{lemma} \label{lem:entropy_ratio} Suppose that for all
  $(n!)^{a_j-1}$ matchings $S_j$ of arity $a_j>1$, the message $m(S_j)$
  is at most $f_j \cdot (a_j-1)\log(n!)$ bits long.  Then
  $\E[|K_{m(S_j)}(S_j)|] \le f_j \cdot n$, where the expectation is
  taken over random choices of the matching $S_j$.
\end{lemma}
  
In other words, if the message $m(S_j)$ contains only a fraction $f_j$
of the bits needed to encode $S_j$, then a server receiving this
message knows only a fraction $f_j$ of the $n$ tuples in $S_j$.  We
will apply the lemma separately to each input $S_j$, for
$j=1,\dots,\ell$ in the next section, for appropriate choices of
$f_j$.

  \begin{proof} 
  Let $m$ be a possible value for $m(S_j)$.
  Since $m$ fixes precisely $|K_m(S_j)|$ tuples of $S_j$,
  \begin{align}
  \log|\setof{S_j}{m(S_j)=m}| & \le (a_j-1) \sum_{i=1}^{n-|K_m(S_j)|} \log i \nonumber\\
  & \le (1-|K_m(S_j)|/n)(a_j-1) \sum_{i=1}^n \log i  \nonumber \\
  &= (1-|K_m(S_j)|/n) \log (n!)^{a_j-1}. \label{eq:ratiobound}
  \end{align}
  We can bound the value we want by considering the binary entropy of the
  distribution $S_j$. By applying the chain rule for entropy, we have
  \begin{align}
  H(S_j) &=H(m(S_j))+ \sum_{m}\P(m(S_j)=m)\cdot H(S_j|m(S_j)=m)\nonumber\\
  &\le f_j \cdot  H(S_j) +  \sum_{m}\P(m(S_j)=m)\cdot H(S_j|m(S_j)=m)\nonumber\\
  &\le f_j \cdot  H(S_j)+
      \sum_{m}\P(m(S_j)=m)\cdot (1-\frac{|K_m(S_j)|}{n}) H(S_j)\nonumber\\
  &= f_j \cdot  H(S_j)+ 
      (1-\sum_{m}\P(m(S_j)=m)\ \frac{|K_m(S_j)|}{n}) H(S_j)\nonumber\\
  &= f_j \cdot H(S_j)+  (1-\frac{\E [|K_{m(S_j)}(S_j)|]}{n}) H(S_j) \label{eq:entropy}
  \end{align}
  where the first inequality follows from the assumed upper bound on $|m(S_j)|$,
  the second inequality follows by \eqref{eq:ratiobound}, and the last two
  lines follow by definition.
  Dividing both sides of \eqref{eq:entropy} by $H(S_j)$ since $H(S_j)$ is not
  zero and rearranging we obtain the required statement.
  \end{proof}

\subsubsection{Bounding the Knowledge of the Query}
\label{subsec:onestepB}

Let $c$ be a constant such that each server receives at most
$c N/p^{1-\varepsilon}$ bits. Let us also fix some server. The message
$m=m(I)$ received by the server is the concatenation of $\ell$ messages, 
one for each input relation. $K_m(S_j)$ depends only on $m(S_j)$, so we 
can assume w.l.o.g. that $K_m(S_j) = K_{m(S_j)}(S_j)$.

In order to represent the total input $I$, we need 
$\sum_{j=1}^{\ell}(a_j-1) \log (n!) = (a- \ell) \log(n!)$ bits. Hence, the message $m(I)$ 
will contain at most $c(a-\ell) \log(n!) / p^{1-\varepsilon}$ bits. Now, for
each relation $S_j$, let us define 
$$f_j =  \frac{\max_{S_j} |m(S_j)|}{(a_j-1) \log(n!)} .$$
Note that this is well-defined, since $a_j > 1$.
Thus, $f_j$ is the largest fraction of bits of $S_j$ that the server
receives, over all choices of the matching $S_j$.  We now derive an
upper bound on the $f_j$'s.  As we discussed before, each part
$m(S_j)$ of the message is constructed independently of the other
relations. Hence, it must be that
$$\sum_{j=1}^{\ell} \max_{S_j} |m(S_j)| \leq c(a-\ell) \log(n!) / p^{1-\varepsilon}$$
By substituting the definition of $f_j$ in this equation, we obtain that 
 $\sum_{j=1}^{\ell} f_j (a_j-1) \leq c (a-\ell)/p^{1-\varepsilon}$.
\autoref{lem:entropy_ratio}  further implies that, for a randomly chosen matching database $I$,
 $\E[|K_{m(I)}(S_j)|]=\E[|K_{m(S_j)}(S_j)|] \leq f_j \cdot n$ for
  all $j\in [\ell]$. We prove:

  \begin{lemma} \label{prop:friedgut} $\mathbf{E}[|K_{m(I)}(q)|] \leq
    g_{q,c} \cdot \mathbf{E}[|q(I)|] / p^{(1-\varepsilon)\tau^*}$ for
    randomly chosen matching database $I$, where $g_{q,c} =
    (c(a-\ell)/\tau^*)^{\tau^*}$ is a constant that depends only on
    the constant $c$ and the query $q$.
\end{lemma}
 
This proves \autoref{th:onestep}, since we can apply a union bound to show
that the total number of tuples known by all $p$ servers is bounded by: 
\begin{align*}
  p \cdot \mathbf{E}[|K_{m(I)}(q)|] 
    \leq p\cdot g_{q,c} \cdot \mathbf{E}[|q(I)|] / p^{(1-\varepsilon)\tau^*}
\end{align*} 
which is the upper bound in \autoref{th:onestep} since $g_{q,c}$
is a constant when $q$ is fixed. 

In the rest of the section we prove \autoref{prop:friedgut}.  
We start with some notation.  For $\ba_j \in [n]^{a_j}$, let
$w_j({\ba_j})$ denote the probability that the server knows the tuple
$\ba_j$.  In other words $w_j(\ba_j) = \P(\ba_j \in K_{m_j(S_j)}(S_j))$, where
the probability is over the random choices of $S_j$.

\begin{lemma} \label{lem:bounds}
For any relation $S_j$ of arity $a_j > 1$:
\begin{itemize}
\item[(a)] $\forall \ba_j \in [n]^{a_j}: w_j(\ba_j) \leq n^{1-a_j}$, and
\item[(b)] $\sum_{\ba_j \in [n]^{a_j}} w_j({\ba_j}) \leq f_j \cdot n$.
\end{itemize}
\end{lemma}

\begin{proof}
  To show (a), notice that $w_j(\ba_j) \leq \P(\ba_j \in S_j) = n^{1-a_j}$,
  while (b) follows from the fact $\sum_{\ba_j \in [n]^{a_j}} w_j({\ba_j}) =
  \mathbf{E}[|K_{m_j(S_j)}(S_j)|] \leq f_j \cdot n$.
\end{proof}

Since the server receives a separate message for each relation
$S_j$, from a distinct input server, the events $\ba_1 \in K_{m_1}(S_1),
\ldots, \ba_\ell \in K_{m_\ell}(S_\ell)$ are independent, hence:
\begin{align*}
  \mathbf{E}[|K_{m(I)}(q)|] = \sum_{\ba \in [n]^k} \P(\ba \in K_{m(I)}(q)) = \sum_{\ba \in [n]^k} \prod_{j=1}^{\ell}
  w_j({\ba_j})
\end{align*}
We now prove \autoref{prop:friedgut} using Friedgut's
inequality. Recall that in order to apply the inequality, we need to 
find a fractional edge cover.  Fix an optimal fractional edge packing
$\mathbf{u} = (u_1, \ldots, u_\ell)$ as in \autoref{tab:LP}.  By
duality, we have that $ \sum_j u_j = \tau^* $,
 where $\tau^*$ is the fractional covering number (which is
 the value of the optimal {\em fractional vertex cover}, and equal to the
 value of the optimal {\em fractional edge
   packing}).  Given $q$, defined as in \eqref{eq:q},
 consider the {\em extended query}, which has a new unary
 atom for each variable $x_i$:
\begin{align*}
  q'(x_1,\ldots, x_k) = S_1(\bar x_1), \ldots, S_\ell(\bar x_\ell),
  T_1(x_1), \ldots, T_k(x_k)
\end{align*}
For each new symbol $T_i$, define 
$  u_i' = 1 - \sum_{j: x_i \in \text{vars}(S_j)} u_j$.
Since $\mathbf{u}$ is a packing,  $u_i' \geq 0$. 
Let us define $\mathbf{u}'=(u_1', \ldots, u_k')$.

\begin{lemma}
(a) The assignment $(\mathbf{u},\mathbf{u}')$ is both a tight
fractional edge packing and a tight fractional edge cover for $q'$. 
(b) $\sum_{j=1}^\ell a_j u_j + \sum_{i=1}^k u'_i = k$
\end{lemma}

\begin{proof}
  (a) is straightforward, since for every variable $x_i$ we have $u_i'+\sum_{j:
    x_i \in \text{vars}(S_j)} u_j = 1$.  Summing up:
  \begin{align*}
k=\sum_{i=1}^k \left( u_i'+\sum_{j: x_i \in \text{vars}(S_j)} u_j \right) = 
\sum_{i=1}^k u_i' + \sum_{j=1}^\ell a_j u_j 
  \end{align*}
which proves (b).  
\end{proof}

We will apply Friedgut's inequality to the extended query $q'$ to prove
\autoref{prop:friedgut}.  Set the variables $w(-)$ used in Friedgut's
inequality as follows:
\begin{align*}
  w_j(\ba_j) = & \P(\ba_j \in K_{m_j(S_j)}(S_j)) \mbox{ for $S_j$,
    tuple $\ba_j \in [n]^{a_j}$} \\
  w_i'(\alpha) = & 1\kern 1.1in  \mbox{ for $T_i$, value $\alpha \in [n]$}
\end{align*}

Recall that, for a tuple $\ba \in [n]^k$ we use $\ba_j \in [n]^{a_j}$
for its projection on the variables in $S_j$; with some abuse, we
write $\ba_i \in [n]$ for the projection on the variable $x_i$.
Assume first that $u_j>0$, for $j=1,\ell$.  Then:
%
\begin{align*}
\allowdisplaybreaks
\mathbf{E}[|K_m(q)|] 
& =  \sum_{\ba \in [n]^k} \prod_{j=1}^{\ell} w_j({\ba_j}) \\
& = \sum_{\ba \in [n]^k} \prod_{j=1}^{\ell} w_j({\ba_j})\prod_{i=1}^k w_i'({\ba_i})\\
&\leq \prod_{j=1}^{\ell} \left( \sum_{\ba \in [n]^{a_j}} w_j({\ba})^{1/u_j} \right)^{u_j}
\prod_{i=1}^{k} \left( \sum_{\alpha \in [n]} w'_i(\alpha)^{1/u_i'} \right)^{u_i'} \\
 &= \prod_{j=1}^{\ell} \left( \sum_{\ba \in [n]^{a_j}} w_j({\ba})^{1/u_j} \right)^{u_j} \prod_{i=1}^k n^{u_i'}
\end{align*}
Note that, since $w'_i(\alpha) = 1$ we have $w'_i(\alpha)^{1/u_i'} =
1$ even if $u_i'=0$.  Write $w_j({\ba})^{1/u_j} = w_j({\ba})^{1/u_j-1}
w_j({\ba})$, and use~\autoref{lem:bounds} to obtain:
\begin{align*}
\sum_{\ba \in [n]^{a_j}} w_j({\ba})^{1/u_j} 
& \leq (n^{1-a_j})^{1/u_j-1} \sum_{\ba \in [n]^{a_j}} w_j({\ba}) \\
& \leq n^{(1-a_j)(1/u_j-1)} f_j \cdot n \\
&  = f_j \cdot n^{( a_j -a_j/u_j + 1/u_j)}
\end{align*}
Plugging this in the bound, we have shown that:
\begin{align}
\mathbf{E}[|K_m(q)|]& \leq \prod_{j=1}^{\ell} (f_j \cdot n^{(a_j -a_j/u_j + 1/u_j)})^{u_j} \cdot \prod_{i=1}^k n^{u_i'} \nonumber\\
  &= \prod_{j=1}^{\ell }f_j^{u_j} \cdot n^{(\sum_{j=1}^{\ell}a_j u_j - a + \ell)} \cdot n^{\sum_{i=1}^k u_i'} \nonumber\\
  &= \prod_{j=1}^{\ell }f_j^{u_j} \cdot n^{(\ell-a)} \cdot n^{(\sum_{j=1}^{\ell}a_j u_j+ \sum_{i=1}^k u_i')} \nonumber\\
  &= \prod_{j=1}^{\ell }f_j^{u_j} \cdot n^{\ell+k-a}
   = \prod_{j=1}^{\ell }f_j^{u_j} \cdot  n^{1+\chi(q)} \nonumber\\
   & = \prod_{j=1}^{\ell }f_j^{u_j} \cdot  \E[|q(I)|]  \label{eq:end}
\end{align}
If some $u_j=0$, then we can derive the same lower bound as follows:
We can replace each $u_j$ with $u_j+\delta$ for any $\delta>0$ still
yielding an edge cover. 
Then we have $\sum_j a_j u_j + \sum_i u_i' = k + a\delta$,
and hence an extra factor $n^{a\delta}$ multiplying the term
$n^{\ell+k-a}$ in \eqref{eq:end}; however, we
obtain the same upper bound since, in the limit as $\delta$ approaches 0,
this extra factor approaches 1.

Let $f_q = \prod_{j=1}^{\ell }f_j^{u_j}$; the final step is to upper bound the quantity $f_q$
using the fact that $\sum_{j=1}^{\ell} f_j (a_j-1) \leq c(a-\ell)/p^{1-\varepsilon}$. Indeed:
\begin{align*}
\allowdisplaybreaks
\log f_q & = \sum_{j=1}^{\ell} u_j \log f_j \\
& = \sum_{j=1}^{\ell} u_j \log \frac{f_j (a_j-1)}{u_j} +  \sum_{j=1}^{\ell} u_j \log \frac{u_j}{a_j-1}  & \\
& \leq \tau^* \sum_{j=1}^{\ell} \frac{u_j}{\tau^*} \log \frac{f_j (a_j-1)}{u_j} \\
& \leq \tau^* \log \frac{\sum_{j=1}^{\ell} f_j(a_j-1)}{\tau^*} \\
& \leq \tau^* \log \frac{c (a-\ell)}{\tau^* p^{1-\varepsilon} } 
\end{align*}
Here, the first inequality comes from the fact that $u_j \leq 1$ and $a_j-1 \geq 1$
(since the arity is at least 2), 
and hence $\log \frac{u_j}{a_j-1}  \leq 0$, for all $j=1, \dots, \ell$. The second inequality follows from Jensen's inequality and concavity of  $\log$. Thus, we obtain
\begin{align}
\label{eq:fq}
f_q = \prod_{j=1}^{\ell} f_j^{u_j} \leq \left( \frac{c (a-\ell)}{\tau^*} \right)^{\tau^*} \cdot
\frac{1}{p^{(1-\varepsilon)\tau^*}}
\end{align}

Recall that we have defined $g_{q,c} = (c(a-\ell)/\tau^*)^{\tau^*}$. Thus, combining 
\eqref{eq:end} with \eqref{eq:fq}  concludes the proof of \autoref{prop:friedgut}.

%
%

\subsection{Extensions}

\autoref{th:onestep:upper} and \autoref{th:onestep} imply that, over
matching databases, the space exponent of a query $q$ is $1-1/\tau^*$,
where $\tau^*$ is its fractional covering number.  \autoref{tab:queries}
illustrates the space exponent for various families of conjunctive
queries.  We now discuss a few extensions and corollaries. 
As a corollary of \autoref{th:onestep} we can 
characterize the queries with space exponent zero, i.e. those that can be
computed in a single round without any replication.

\begin{corollary} \label{cor:root:variable} A query $q$ has
  covering number $\tau^*(q)=1$ if and only if there exists a variable shared
  by all atoms.
\end{corollary} 

\begin{proof}
  The ``if'' direction is straightforward: if $x_i$ occurs in all
  atoms, then $v_i=1$ and $v_j=0$ for all $j\neq i$ is a fractional
  vertex cover with value 1, proving $\tau^*=1$ since $\tau^*(q) \geq
  1$ for any $q$.  
  
  For ``only if'', assume $\tau^*(q)=1$, and consider
  a fractional vertex cover for which $\sum_{i=1}^k v_i = 1$.  We
  prove that there exists a variable $x$ that occurs in all atoms.  If
  not, then every variable $x_j$ is missing from at least one atom
  $S$: since $\sum_{i: x_i \in \text{vars}(S)} v_i \geq 1$, it follows
  that $v_j = 0$, for all variables $v_j$, which is a contradiction.
\end{proof}

Thus, a query can be computed in one round on \mpc(0)
 if and only if it has a variable occurring in all atoms.  
The corollary should be
contrasted with the results in~\cite{DBLP:conf/pods/KoutrisS11}, which
proved that a query is computable in one round iff it is {\em
  tall-flat}.  Any connected tall-flat query has a variable occurring
in all atoms, but the converse is not true in general. The algorithm
in~\cite{DBLP:conf/pods/KoutrisS11} works for {\em any} input data,
including skewed inputs, while here we restrict to matching databases.
For example, $S_1(x,y),S_2(x,y),S_3(x,z)$ can be computed in one round
if all inputs are permutations, but it is not tall-flat, and hence it
cannot be computed in one round on general input data.

\autoref{th:onestep} tells us that a query $q$ can report at most a
$1/p^{\tau^*(q)(1-\varepsilon)-1}$ fraction of answers.  We show 
that there is an algorithm achieving this for matching databases:

\begin{proposition} \label{prop:match:bound} Given a connected query $q$ and
  $\varepsilon < 1 - 1/\tau^*(q)$, there exists an algorithm that reports 
  $\Theta(\E[|q(I)|]/p^{\tau^*(q)(1-\varepsilon)-1})$ answers in expectation
  on any matching database in one round in \mpc$(\varepsilon)$.
\end{proposition}

 \begin{proof} 
 The algorithm we describe here
 is similar to the HC algorithm described in Section~\ref{subsec:algorithms}. 
 For each variable $x_i$, define the shares $p_i = p^{(1-\varepsilon)v_i}$,
 where $\mathbf{v} = (v_1, \dots, v_k)$ is the optimal fractional vertex cover.
 We use random hash functions $h_i: [n] \rightarrow [p_i]$ for each $i=1, \dots,k$.
 This creates $p^{(1-\varepsilon)\tau^*(q)}$ hashing buckets in the 
 $k$-dimensional hypercube $[n]^{p_1} \times \dots \times [n]^{p_k}$.
 Notice that, since $\varepsilon < 1-1/\tau^*(q)$, the number of points in the
 hypercube is strictly greater than $p$, so it is not possible to assign each point to one of the
 $p$ servers. Instead, we will pick $p$ points uniformly at random and assign each of
 the $p$ servers to one of the points in the hypercube. Since each potential output 
 tuple is hashed to a random point of the hypercube, the probability that a potential output
 tuple is covered by one of the servers is $p/ p^{(1-\varepsilon)\tau^*(q)}$. Our algorithm will 
 execute exactly as the HC algorithm, but will communicate only tuples that
 are hashed to one of the chosen hypercube points (and thus to one of the servers).
 
 By our previous discussion, the algorithm reports in expectation 
 $p^{1-(1-\varepsilon)\tau^*(q)} \E[|q(I)|]$ tuples.
To conclude the proof, it suffices to show that each of the servers will
 receive $O(n/p^{1-\varepsilon})$ tuples. Indeed, using a similar analysis
 to the HC algorithm, each server receives from a relation $S_j$  in 
 expectation $n/\prod_{i: x_i \in \text{vars}(S_j)} p_i$ tuples. 
Since $\mathbf{v}$ is a fractional cover, we have that
 $\prod_{i: x_i \in \text{vars}(S_j)} p_i = 
 p^{(1-\varepsilon) \sum_{i:x_i \in \text{vars}(S_j)}v_i } \geq p^{1-\varepsilon}$.
Thus, the number of tuples received by each server is in expectation $O(n/p^{1-\varepsilon})$.
Moreover, following the same argument as in the proof of \autoref{th:onestep:upper},
the probability that the number of tuples per server deviates more than a constant
factor from the expectation is exponentially small to $n$.
 \end{proof}

Note that the algorithm is forced to run in one round,
in an \mpc($\varepsilon$) model strictly weaker than its space
exponent, hence it cannot find all the answers: the proposition says
that the algorithm can find an expected number of answers that matches
\autoref{th:onestep}.

So far, our lower bounds were for the {\sc Join-Reporting} problem.
We can extend the lower bounds to the {\sc Join-Witness}
problem. 

\begin{proposition}
For $\varepsilon < 1/2$, there exists no one-round \mpc($\varepsilon$) 
algorithm that solves {\sc Join-Witness} for the query $q(w,x,y,z) = 
R(w),S_1(w,x),S_2(x,y),S_3(y,z),T(z)$.
\end{proposition}

\begin{proof}
Consider the family of inputs where $S_1, S_2, S_3$ are 2-dimensional
matchings, while $R,T$ are uniformly at random chosen subsets of $[n]$ 
of size $\sqrt{n}$. It is easy to see that for this input distribution $I$,
$\E[|q(I)|] = 1$, since the probability that each answer of the subquery 
$q' = S_1, S_2, S_3$ is included in the final output is 
$(1/\sqrt{n}) \cdot (1/\sqrt{n}) = 1/n$,
while $q'$ has exactly $n$ answers.

Since the size of $R,T$ is $\sqrt{n}$ and $n \gg p$, we can assume w.l.o.g. that both
relations are broadcast to all servers during the first round. Further, we can assume
w.l.o.g. as before that $S_1, S_2, S_3$ are initially stored in three separate input servers.

Fix some server $s$; by \autoref{th:onestep}, in expectation the server containsonly 
$O(\E[|q'(I)|]/p^{2(1-\varepsilon)})$ tuples from the subquery $q'$. Since $R,T$ are known to the
server, and $\E[|q'(I)|] = n$, the server will know in expectation for 
$(1/\sqrt{n}) \cdot (1/\sqrt{n}) O(n/p^{2(1-\varepsilon)}) = O(1/p^{2(1-\varepsilon)})$ tuples from $q$. 
Consequently, the servers know in total for $O(1/p^{2(1-\varepsilon)-1})$ output tuples in expectation. 
Since $\varepsilon < 1/2$, and the input has in expectation one answer, we can argue as in 
\autoref{cor:yao} that the probability that any algorithm will provide the tuple as a witness is
polynomially small.
\end{proof}

%% file: multistep.tex
\section{Multiple Communication Steps}
\label{sec:multistep}

In this section we consider a restricted version of 
the \mpc($\varepsilon$) model, called the {\em tuple-based}
\mpc($\varepsilon$) model, which can simulate multi-round MapReduce for
database queries. We will establish both upper and lower bounds on 
the number of rounds needed to compute any connected query $q$ in
this tuple-based \mpc($\varepsilon$) model, proving \autoref{th:intro2}.

\subsection{An Algorithm for Multiple Rounds}

Given an $\varepsilon \geq 0$, let $\Gamma_\varepsilon^1$ denote the
class of connected queries $q$ for which $\tau^*(q) \leq
1/(1-\varepsilon)$; these are precisely the queries that can be
computed in one round in the \mpc($\varepsilon$) model on matching
databases.  We extend this definition inductively to larger numbers of
rounds: Given $\Gamma^r_\varepsilon$ for some $r\ge 1$, define
$\Gamma_\varepsilon^{r+1}$ to be the set of all connected queries $q$
constructed as follows.  Let $q_1, \ldots, q_m \in
\Gamma_\varepsilon^r$ be $m$ queries, and let $q_0 \in
\Gamma_\varepsilon^1$ be a query over a different vocabulary $V_1,
\ldots, V_m$, such that $|\text{vars}(q_j)|=\text{arity}(V_j)$ for all
$j\in [m]$.  Then, the query $q = q_0[q_1/V_1, \ldots, q_m/V_m]$,
obtained by substituting each view $V_j$ in $q_0$ with its definition
$q_j$, is in $\Gamma_\varepsilon^{r+1}$.  In other words,
$\Gamma^r_\varepsilon$ consists of queries that have a {\em query
  plan} of depth $r$, where each operator is a query computable in one
step.

The following proposition is straightforward.
 
\begin{proposition} \label{prop:multistep} 
Every query in
  $\Gamma_\varepsilon^r$ can be computed by an \mpc($\varepsilon$) algorithm
  in $r$ rounds on any matching database.
\end{proposition}

\begin{example} \label{ex:queryplan} Let $\varepsilon=1/2$. The
  query $L_k$ in \autoref{tab:queries} for $k=16$ has a query plan
  of depth $r=2$.  The first step computes in parallel four queries,
  $v_1 = S_1, S_2, S_3, S_4$, \ldots, $v_4 = S_{13}, S_{14}, S_{15},
  S_{16}$.  Each is isomorphic to $L_4$, therefore
  $\tau^*(q_1)=\cdots=\tau^*(q_4) = 2$ and each can be computed in one
  step.  The second step computes the query $q_0 = V_1,
  V_2, V_3, V_4$, which is also isomorphic to $L_4$.  We can
  generalize this approach for any $L_k$:
  for any $\varepsilon \geq 0$, let $k_{\varepsilon}$ be the
  largest integer such that $\tau^*(L_{k_{\varepsilon}}) \leq
  1/(1-\varepsilon)$: $k_{\varepsilon} = 2 \lfloor 1/(1-\varepsilon)
  \rfloor$.  Then, for any $k \geq k_{\varepsilon}$, $L_k$ can be
  computed using $L_{k_{\varepsilon}}$ as a building block at each
  round: the plan will have a depth of $\lceil \log k / \log
  k_{\varepsilon} \rceil$.

  We also consider the query $SP_k =
  \bigwedge_{i=1}^k R_i(z,x_i), S_i(x_i, y_i)$. Since $\tau^*(SP_k) =
  k$, the space exponent for one round is $1-1/k$. However, $SP_k$ has
  a query plan of depth 2 for \mpc(0), by computing the joins $q_i =
  R_i(z,x_i), S_i(x_i, y_i)$ in the first round and in the second
  round joining all $q_i$ on the common variable $z$.  Thus, if we
  insist in answering $SP_k$ in one round, we need a
  huge replication $O(p^{1-1/k})$, but we can compute it in two rounds
  with replication $O(1)$.
\end{example}

We next present an upper bound on the number of rounds needed to compute any
query.  Let $\text{rad}(q) = \min_u \max_v d(u,v)$ denote the {\em
  radius} of a query $q$, where $d(u,v)$ denotes the distance between
two nodes in the hypergraph.  For example, $\text{rad}(L_k) = \lceil
k/2 \rceil$ and $\text{rad}(C_k) = \lfloor k/2 \rfloor$.

\begin{sloppypar}
\begin{lemma} 
\label{lem:multitree}
Fix $\varepsilon \geq 0$, let $k_{\varepsilon} = 2 \lfloor
1/(1-\varepsilon) \rfloor$, and let $q$ be any connected query.  Let
$r(q)=\lceil \log (\text{rad}(q)) / \log k_{\varepsilon} \rceil +1$ if
$q$ is tree-like, and let $r(q)=\lceil \log (\text{rad}(q)+1) / \log
k_{\varepsilon} \rceil +1$ otherwise.  Then, $q$ can be computed in
$r(q)$ rounds on any matching database input by repeated application
of the HC algorithm in the \mpc($\varepsilon$) model.
\end{lemma}
\end{sloppypar}

\begin{proof}
By definition of $\text{rad}(q)$, there exists some node
$v \in \text{vars}(q)$,
such that the maximum distance of $v$ to any other node in the hypergraph of
$q$ is at most $\text{rad}(q)$. 
If $q$ is tree-like then we can decompose $q$ into a set of at most
$|\text{atoms}(q)|^{\text{rad}(q)}$
(possibly overlapping) paths ${\cal P}$ of length $\leq \text{rad}(q)$,
each having $v$ as one endpoint.
Since it is essentially isomorphic to $L_{\ell}$,
a path of length $\ell \leq \text{rad}(q)$ can be computed 
in at most  $\lceil \log (\text{rad}(q)) / \log k_{\varepsilon} \rceil$
rounds using the query plan from \autoref{prop:multistep} together with 
repeated use of the one-round HC algorithm for paths of length $k_\varepsilon$
as shown in \autoref{th:onestep:upper} for $\tau=1/(1-\varepsilon)$.
Moreover, all the paths in ${\cal P}$ can be computed in parallel, because
 $|{\cal P}|$ is a constant depending only on $q$.
Since every path will contain variable $v$, we can compute the join of
all the paths in one final round without any replication.
The only difference for general connected queries is that $q$ may also contain
atoms that join vertices at distance $\text{rad}(q)$ from $v$ that are
not on any of the paths of length $\text{rad}(q)$ from $v$: these
can be covered using paths of length $\text{rad}(q)+1$ from $v$.
\end{proof}

As an application of this proposition, \autoref{tab:complexity} shows
the number of  rounds required by different types of
queries.

\begin{table}
  \centering
  \begin{tabular}{|c|c|c|c|} \hline
  q & $\varepsilon$           & $r$                     & $r = f(\varepsilon)$ \\
    query    & space exponent        & rounds for $\varepsilon=0$           & rounds/space tradeoff\\ \hline

$C_k$ & $1-2/k$ & $\lceil \log k \rceil$ & $\sim \frac{\log k}{\log(2/(1-\varepsilon))}$ \\ \hline
$L_k$ & $1-\frac{1}{\lceil k/2 \rceil}$ & $\lceil \log k \rceil$ & $\sim \frac{\log k}{\log(2/(1-\varepsilon))}$ \\ \hline
$T_k$ & $0$ & 1 & NA \\ \hline
$SP_k$ & $1-1/k$ & 2 & NA \\ \hline
  \end{tabular}

  \caption{The tradeoff between space and communication rounds for several queries.}
  \label{tab:complexity}
\end{table}

\subsection{Lower Bounds for Multiple Rounds}

Our lower bound results for multiple rounds are restricted in two
ways: they apply only to an \mpc\ model where communication at
rounds $\geq 2$ is of a restricted form, and they match the upper
bounds only for a restricted class of queries.

\subsubsection{Tuple-Based \mpc} 

\label{subsec:multiround:upper}

Recall that $M^1_u = (M^1_{1u}, \ldots,
M^1_{\ell u})$, where $M^1_{ju}$ denotes the message sent during round
1 by the input server for $S_j$ to the worker $u$.  Let $I$ be the
input database instance, and $q$ be the query we want to compute.  A
{\em join tuple} is any tuple in $q'(I)$, where $q'$ is
any connected subquery of $q$.

The {\em tuple-based} \mpc($\varepsilon$) model imposes the following
two restrictions during rounds $r \geq 2$, for every worker $u$: (a)
the message $M^r_{uv}$ sent to $v$ is a set of join tuples, and (b)
for every join tuple $t$, the worker $u$ decides whether to include
$t$ in $M^r_{uv}$ based only on $t, u, v, r$ and $M^1_{ju}$, for all
$j$ s.t. $t$ contains a base tuple in $S_j$.

The restricted model still allows unrestricted communication during
the first round; the information $M^1_u$ received by server $u$ in the
first round is available throughout the computation.  However, during
the following rounds, server $u$ can only send messages consisting of
join tuples, and, moreover, the destination of these join tuples
can depend only on the tuple itself and on $M^1_u$.  Since a join tuple
is represented using $\Theta(\log n)$ bits, each server receives
$O(n/p^{1-\varepsilon})$ join tuples at each round.  
For convenience, when we have fixed the constant $c$ in the bound on the 
number bits or tuples received by each processor at each step, we refer
to the algorithm as a tuple-based \mpc($\varepsilon,c$) algorithm.

The restriction of communication to join tuples (except for the first round
during which arbitrary (e.g., statistical) information can be sent) is 
natural and the tuple-based \mpc\ model captures a wide variety of algorithms
including those based on MapReduce.
Since the servers can perform arbitrary inferences
based on the messages that they receive, even a limitation to messages that
are join tuples starting in the second round, without a restriction on how
they are routed, would still essentially have been equivalent to the fully
general \mpc\ model: For example, any server wishing to send a sequence of
bits to
another server can encode the bits using a sequence of tuples that the two
exchanged in previous rounds, or (with slight loss in efficiency) using the 
understanding that the tuples themselves are not important, but some
arbitrary fixed Boolean function of those tuples is the true message being
communicated.  This explains the need for the condition on routing tuples
that the tuple-based \mpc\ model imposes.

We now describe 
the lower bound for multiple rounds in the tuple-based \mpc\ model.

\subsubsection{A Lower Bound}

We give here a general lower bound for connected, conjunctive queries,
and show how to apply it to $L_k$, to tree-like queries, and to $C_k$;
these results prove \autoref{th:intro2}.  We postpone the proof to the
next subsection.

\begin{definition} \label{def:m} Let $q$ be a connected, conjunctive
  query.  A set $M\subseteq\text{atoms}(q)$ is {\em $\varepsilon$-good
    for $q$} if it satisfies:
    \begin{packed_enum}
    \item Every subquery of $q$ that is in $\Gamma_{\varepsilon}^1$
      contains at most one atom in $M$. ($\Gamma_{\varepsilon}^1$
      defined in Sec.~\ref{subsec:multiround:upper})
    \item $\chi(\contracted) = 0$, where $\contracted =
      \text{atoms}(q) - M$.  (Hence by \autoref{lemma:chi}, $\chi(q/
      \contracted)=\chi(q)$.  This condition is equivalent to each
      connected component of $\contracted$ being tree-like.)
    \end{packed_enum}
    An {\em $(\varepsilon,r)$-plan} $\cal M$ is a sequence
    $M_1,\ldots, M_r$, with $M_0=\text{atoms}(q)\supset M_1
    \supset \cdots M_r$ such that
    (a) for all $j\in [r]$, $M_{j+1}$ is $\varepsilon$-good for
      $q/\contracted_j$ where $\contracted_j=\text{atoms}(q)-M_j$, and
    (b) $q/\contracted_r\;\notin \Gamma^1_\varepsilon$.
\end{definition}

\begin{theorem} \label{th:multiround} If $q$ has a
  $(\varepsilon,r)$-plan then every randomized algorithm
  running in $r+1$ rounds on the tuple-based \mpc($\varepsilon$) model
  with $p=\omega(1)$ processors fails to compute $q$ with probability
  $\Omega(n^{\chi(q)})$.
\end{theorem}

We prove the theorem in the next section.  Here, we show how to apply
it to three cases.  Assume $p=\omega(1)$, and recall that
$k_{\varepsilon} = 2 \lfloor 1/(1-\varepsilon) \rfloor$
(\autoref{ex:queryplan}).  First, consider $L_k$.

\begin{lemma} \label{lemma:lk:lower} Any tuple-based
  \mpc($\varepsilon$) algorithm that computes $L_k$ needs at least $
  \lceil \log k / \log k_{\varepsilon} \rceil$ rounds.
\end{lemma}

\begin{proof}
  We show inductively how to produce an $(\varepsilon,r)$-plan for
  $L_k$ with $r=\lceil \log k / \log k_{\varepsilon} \rceil-1$.
  Recall that $\Gamma^1_\varepsilon$ consists of connected queries for
  which $\tau^*(q) \leq 1/(1-\varepsilon)$: thus, the subqueries of
  $L_k$ that are in $\Gamma^1_\varepsilon$ are precisely queries of
  the form $S_j(x_{j-1},x_j),S_{j+1}(x_j,x_{j+1}),\ldots,
  S_{j+k_0-1}(x_{j+k_0-2},x_{j+k_0-1})$, in other words they are
  isomorphic to $L_{k_0}$, where $k_0\leq k_\varepsilon$.  Therefore,
  we obtain an $\varepsilon$-good set $M$ for $L_\ell$ if we include
  every $k_{\varepsilon}$-th atom in $L_\ell$, starting with the first
  atom: $S_1, S_{k_\varepsilon+1}, S_{2k_\varepsilon+1},\ldots$ Then
  $L_k/\contracted_1 = S_1(x_0,x_1), S_{k_\varepsilon+1}(x_1,
  x_{k_\varepsilon+1}), S_{2k_\varepsilon+1}(x_{k_\varepsilon+1},
  x_{2k_\varepsilon+1}), \ldots$ is isomorphic to $L_{\lceil
    k/k_{\varepsilon} \rceil}$.  Similarly, for $j=2,..,r$, choose
  $M_j$ to consist of every $k_{\varepsilon}$-th atom starting at the
  first atom in $L_k/\contracted_{j-1}$.  Finally,
  $L_k/\contracted_{j-1}$ will be isomorphic to a path query of length
  $L_\ell$ for some $\ell\ge k_{\varepsilon} +1$ and hence is not in
  $\Gamma^1_\epsilon$.  Thus $M_1,\ldots, M_r$ is the desired
  $(\varepsilon,r)$-plan and the lower bound follows from
  \autoref{th:multiround}.
\end{proof}

Combined with \autoref{ex:queryplan}, it implies that $L_k$ requires
precisely $ \lceil \log k / \log k_{\varepsilon} \rceil$ rounds on the
tuple-based \mpc($\varepsilon$).

Second, we give a lower bound for tree-like queries, and for that we
use a simple observation:
\begin{proposition} \label{prop:tree-like}
If $q$ is a tree-like query, and $q'$ is any connected subquery of $q$, $q'$ needs
at least as many rounds as $q$ in the tuple-based \mpc($\varepsilon$) model. 
\end{proposition}
\begin{proof}
  Given any tuple-based \mpc($\varepsilon$) algorithm $A$ for
  computing $q$ in $r$ rounds we construct a tuple-based
  \mpc($\varepsilon$) algorithm $A'$ that computes $q'$ in $r$ rounds.
  $A'$ will interpret each instance over $q'$ as part of
  an instance for $q$ by using the relations in $q'$ and using the
  identity permutation ($S_j =
  \set{(1,1,\ldots),(2,2,\ldots),\ldots}$) for each relation in $q
  \setminus q'$.  Then, $A'$ runs exactly as $A$ for $r$ rounds; after
  the final round, $A'$ projects out for every tuple all the variables
  not in $q'$.  The correctness of $A'$ follows from the fact that $q$
  is tree-like.
\end{proof}

Define $\text{diam}(q)$, the {\em diameter} of a query $q$, to be the
longest distance between any two nodes in the hypergraph of $q$.  In
general, $\text{rad}(q) \leq \text{diam}(q) \leq 2\ \text{rad}(q)$.
For example, $\text{rad}(L_k) = \lfloor k/2 \rfloor$,
$\text{diam}(L_k)=k$ and $\text{rad}(C_k) = \text{diam}(C_k)=\lfloor
k/2 \rfloor$. \autoref{lemma:lk:lower} and \autoref{prop:tree-like}
imply:

\begin{sloppypar}
\begin{corollary}
  Any tuple-based \mpc($\varepsilon$) algorithm that computes a
  tree-like query $q$ needs at least 
  $\lceil \log_{k_{\varepsilon}} (\text{diam}(q)) \rceil$ rounds.
\end{corollary}
\end{sloppypar}

\begin{sloppypar}
  Let us compare the lower bound $r_{\text{low}}=\lceil
  \log_{k_{\varepsilon}} (\text{diam}(q)) \rceil$ and the upper bound
  $r_{\text{up}}=\lceil \log_{k_{\varepsilon}} (\text{rad}(q)) \rceil
  +1$ (\autoref{lem:multitree}): $\text{diam}(q) \leq 2\text{rad}(q)$
  implies $r_{\text{low}} \leq r_{\text{up}}$, while $\text{rad}(q)
  \leq \text{diam}(q)$ implies $r_{\text{up}} \leq r_{\text{low}} +1$.
  The gap between the lower bound and the upper bound is at most 1,
  proving \autoref{th:intro2}.  When $\varepsilon < 1/2$, these bounds
  are matching, since $k_\varepsilon = 2$ and $2\text{rad}(q)-1\leq
  \text{diam}(q)$ for tree-like queries.  The tradeoff between the
  space exponent $\varepsilon$ and the number of rounds $r$ for
  tree-like queries is $r \cdot \log \frac{2}{1-\varepsilon} \approx
  \log (\text{rad}(q))$.
\end{sloppypar}

Third, we study one instance of a non tree-like query:

\begin{lemma}
  Any tuple-based \mpc($\varepsilon$) algorithm that computes the query $C_k$
  needs at least $ \lceil \log (k/(m_{\varepsilon}+1)) / \log
  k_{\varepsilon} \rceil + 1$ rounds, where $m_{\varepsilon} = \lfloor
  2/(1-\varepsilon)\rfloor$.
\end{lemma}

\begin{proof}
  Observe that any set $M$ of atoms that are (at least)
  $k_\varepsilon$ apart along any cycle $C_\ell$ is $\epsilon$-good
  for $C_\ell$ and $C_\ell/\contracted$ is isomorphic to $C_{\lfloor
    \ell/k_\varepsilon\rfloor}$.  If $k\ge k_{\varepsilon}^r
  (m_\varepsilon+1)$, we can repeatedly choose such $\varepsilon$-good
  sets to construct an $(\varepsilon,r)$-plan $M_1,\ldots, M_r$ such
  that the final contracted query $C_k/\contracted_r$ contains a cycle
  $C_{\ell'}$ with $\ell'\ge m_{\varepsilon} +1$ (and therefore cannot
  be computed in 1 round by any \mpc($\varepsilon$) algorithm).  The
  result now follows from \autoref{th:multiround}.
\end{proof}

Here, too, we have a gap of 1 between this lower bound and the upper
bound in \autoref{lem:multitree}. Consider $C_5$ and $\varepsilon=0$;
$\text{rad}(C_5)=\text{diam}(C_5)=2$,
$k_\varepsilon=m_\varepsilon=2$. The lower bound is $\lfloor \log 5/3
\rfloor + 1 = 2$ rounds, the upper bound is $\lceil \log 3 \rceil + 1
= 3$ round.  The exact number of rounds for $C_5$ is open.

As a final application, we show how to apply \autoref{lemma:lk:lower}
to show that transitive closure requires many rounds. In particular, we consider
the problem \textsc{Connected-Components}, for which, given an undirected graph 
$G = (V,E)$ the requirement is to label the nodes of each connected component
with the same label, unique to that component.

\begin{theorem}
\label{th:connected-comps}
  For any fixed $\varepsilon < 1$, there is no $p$-server algorithm in the
  tuple-based \mpc($\varepsilon$) model that uses $o(\log p)$ rounds
  and computes \textsc{Connected-Components} on an arbitrary input graph.
\end{theorem}

The basic idea of the proof of this theorem is to construct input graphs
for \textsc{Connected-Components} whose components correspond to the output
tuples for $L_k$ for $k=p^\delta$ for some small constant $\delta$ depending
on $\varepsilon$ and use the round lower bound for solving $L_k$.
In this instance, the size of the query $L_k$ is not fixed, but 
depends on the number of processors $p$.  The lower bound in
\autoref{th:multiround} does not apply in this case but in the next section
we will prove a more precise and general result, 
\autoref{th:strong-multiround},
from which we can derive both \autoref{th:multiround} and
\autoref{th:connected-comps}.

\subsubsection{Proofs of Theorems~\ref{th:multiround} and~\ref{th:connected-comps}}

\begin{sloppypar}
Given an $(\varepsilon,r)$-plan $\cal M$ (\autoref{def:m}) for a
query $q$, define $\tau^*(\cal M)$ to be the minimum of
$\tau^*(q/\contracted_r)$, and the minimum of $\tau^*(q')$, where $q'$
ranges over all connected subqueries of $q/\contracted_{j-1}$, $j\in
[r]$, such that $q' \not\in \Gamma^1_\varepsilon$.  Since every $q'$
satisfies $\tau^*(q')(1-\varepsilon)> 1$ (by $q' \not\in
\Gamma^1_\varepsilon$), and $\tau^*(q/\contracted_r)
(1-\varepsilon)>1$ (by the definition of goodness), we have
$\tau^*(\mathcal{M})(1-\varepsilon)>1$.
Further, define the set 
\begin{align*}
\mathcal{S}(q) = \setof{q'}{q' \notin \Gamma^1_{\varepsilon},\  
q' \text{ is a minimal connected subquery of } q}.
\end{align*}
and let
\begin{align*}
\gamma_{c}(q,\mathcal{M})=
g_{q/\contracted_r,c(r+1)} +
\sum_{j=1}^{r} \sum_{q' \in \mathcal{S}(q/\contracted_{j-1})} g_{q',c(r+1)} 
\end{align*}
where $g_{q',c'}=(c'(a(q')-\ell(q'))/\tau^*(q'))^{\tau^*(q')}$ is the constant
defined in \autoref{prop:friedgut}, $a(q')$ is the total arity of $q'$, and
$\ell(q')$ is the number of atoms in $q'$.
\end{sloppypar}

\begin{theorem} \label{th:strong-multiround} 
  If $q$ has an $(\varepsilon,r)$-plan $\cal M$ then 
  any deterministic tuple-based \mpc($\varepsilon,c$) 
  algorithm running in $r+1$ rounds
  reports at most $\gamma_c(q,\mathcal{M}) \cdot \mathbf{E}[|q(I)|]/p^{\tau^*(\mathcal{M})(1-\varepsilon)-1}$
  correct answers in expectation over uniformly chosen matching database $I$.
\end{theorem}

Observe that for a constant-sized query $q$ and constant $r$,
$\gamma_c(q,\mathcal{M})$ is a constant.
The argument in \autoref{cor:yao} then extends immediately to this case,
implying that every randomized tuple-based \mpc($\varepsilon$)
algorithm with $p=\omega(1)$ and $r+1$ rounds will fail to compute
$q$ with probability $\Omega(n^{\chi(q)})$.  This proves
\autoref{th:multiround}.

The rest of this section gives the proof of this theorem.  The
intuition is this.  Consider a $\varepsilon$-good set $M$; then any
matching database $i$ consists of two parts, $i=(i_M, i_\contracted)$,
where $i_M$ are the relations for atoms in $M$, and $i_\contracted$
are the other relations. We show that, for a fixed instance
$i_\contracted$, the algorithm $A$ can be used to compute
$q/\contracted(i_M)$ in $r+1$ rounds; however, the first round is almost
useless, because the algorithm can discover only a tiny number of join
tuples with two or more atoms $S_j \in M$, since every subquery $q'$
of $q$ that has two $M$-atoms is not in $\Gamma^1_\varepsilon$.  This
shows that the algorithm computes $q/\contracted(i_M)$ in only $r$
rounds, and we repeat the argument until a one-round algorithm remains.

First, we need some notation.  For a connected subquery $q'$ of $q$,
$q'(I)$ denotes as usual the answer to $q'$ on an instance $I$.
Whenever $\text{atoms}(q') \subseteq \text{atoms}(q'')$, then we say
that a tuple $t'' \in q''(I)$ {\em contains} a tuple $t' \in q'(I)$,
if $t'$ is equal to the projection of $t''$ on the variables of $q'$;
if $A \subseteq q''(I), B \subseteq q'(I)$, then $A \ltimes B$, called
the {\em semijoin}, denotes the subset of tuples $t'' \in A$ that
contain some tuple $t' \in B$.

Let $A$ be a deterministic algorithm with $r+1$ rounds, $k\in [r+1]$ a
round number, $u$ a server, and $q'$ a subquery of $q$.  For a
matching database input $i$, define $m_{A,u,k}(i)$ to be the vector of
messages received by server $u$ during the first $k$ rounds of the
execution of $A$ on input $i$.  Define $m_{A,k}(i)=(m_1,\ldots, m_p)$,
where $m_u=m_{A,u,k}(i)$ for all $u\in [p]$, and:
\begin{align*}
    K^{A,u,k}_m(q') & =  \{t' \in [n]^{\text{vars}(q')}\mid \forall \mbox{ matching databases }i,
      m_{A,u,k}(i)=m \Rightarrow t'\in q'(i)\}\\
    K^{A,k}_m(q') & = \bigcup_u K^{A,u,k}_{m_u}(q') \\
    A(i) & =K^{A,r+1}_{m_{A,r+1}(i)}(q).
\end{align*}
$K^{A,u,k}_{m_{A,u,k}(i)}(q')$ and $K^{A,k}_{m_{A,k}(i)}(q')$ denote the set
of join tuples from $q'$ known at round $k$ by server $u$, and by all servers,
respectively, on input $i$.
$A(i)$ is w.l.o.g.\ the final answer of $A$ on input $i$.
Define 
\begin{align*}
J^{A,q}(i)&= \bigcup \{K^{A,1}_{m_{A,1}(i)}(q')\mid q'\mbox{ connected subquery of }q\} \\
J^{A,q}_\varepsilon(i)&= \bigcup \{K^{A,1}_{m_{A,1}(i)}(q')\mid q'\notin \Gamma^1_\varepsilon\mbox{ connected subquery of }q\}
\end{align*}
$J^{A,q}_\varepsilon(i)$ is precisely the set of join tuples known
after the first round, but which correspond to subqueries that are
themselves not computable in one round; thus, the number of tuples in
$J^{A,q}_\varepsilon(i)$ will be small.  Next, we need two lemmas.

 To prove \autoref{th:strong-multiround}, we need two lemmas.

\begin{lemma} \label{lemma:contraction} Let $q$ be a query, and $M$ be
  any $\varepsilon$-good set for $q$.  If $A$ is an algorithm with
  $r+1$ rounds for $q$, then for any matching database $i_\contracted$
  over the atoms of $\contracted$, there exists an algorithm $A'$ with
  $r$ rounds for $q/\contracted$ 
  using the same number of processors and the same total number of bits
  of communication
  received per processor such that, for every matching
  database $i_M$ defined over the atoms of $M$:
  \begin{align*}
    |A(i_M, i_\contracted)|\le |q(i_M,i_\contracted) \ltimes
    J^{A,q}_\varepsilon(i_M,i_\contracted)|+ |A'(i_M)|.
  \end{align*}
\end{lemma}

In other words, the algorithm returns no more answers than the (very
few) tuples in $J$, plus what another algorithm $A'$ (to be defined)
computes for $q/\contracted$ in {\em one fewer} round.

\begin{proof} The proof requires two constructions.

  {\em 1. Contraction.} Call $q/\contracted$ the {\em contracted}
  query.  While the original query $q$ takes as input the complete
  database $i = (i_M, i_\contracted)$, the input to the contracted
  query is only $i_M$.  We show how to use the algorithm $A$ for $q$
  to derive an algorithm, denoted $A_M$, for $q/\contracted$.

  For each connected component $C$ of $\contracted$, choose a
  representative variable $z_c\in\text{vars}(C)$; also denote $S_C$
  the result of applying the query $C$ to $i_\contracted$; $S_c$ is a
  matching, because $C$ is tree-like.  Denote $\bar \sigma =
  \setof{\sigma_x}{x \in \text{vars}(q)}$, where, for every variable
  $x \in \text{vars}(q)$, $\sigma_x$ is the following permutation on
  $[n]$: if $x \not \in \text{vars}(\contracted)$ then $\sigma_x=$ the
  identity; otherwise $\sigma_x = \Pi_{xz_c}(S_C)$, for the unique
  connected component s.t. $x \in \text{vars}(C)$.  We think of $\bar
  \sigma$ as permuting the domain of each attribute $x \in
  \text{vars}(q)$.  Then $\bar \sigma(q(i)) = q(\bar \sigma(i))$, and
  $\bar \sigma(i_\contracted) = \mathbf{id}_\contracted$ the identity
  matching database (where each relation in $\contracted$ is
  $\set{(1,1,\ldots),(2,2,\ldots),\ldots}$), and therefore:
  \begin{align*}
    q/\contracted(i_M) = &\bar \sigma^{-1}(\Pi_{\text{vars}(q/\contracted)}(q(\bar \sigma(i_M),\mathbf{id}_\contracted)))
  \end{align*}
  (We assume $\text{vars}(q/\contracted) \subseteq \text{vars}(q)$;
  for that, when we contract a set of nodes of the hypergraph, we
  replace them with one of the nodes in the set.)

  The algorithm $A_M$ for $q/\contracted(i_M)$ is this.  First, each
  input server for $S_j \in M$ replaces $S_j$ with $\bar \sigma(S_j)$
  (since $i_\contracted$ is fixed, it is known to all servers, hence,
  so is $\bar \sigma$); next, run $A$ unchanged, substituting all
  relations $S_j \in \contracted$ with the identity; finally, apply
  $\bar \sigma^{-1}$ to the answers and return them.  We have:
    \begin{align}
      A_M(i_M) = \bar \sigma^{-1}(\Pi_{\text{vars}(q/\contracted)}(A(\bar
      \sigma(i_M), \mathbf{id}_\contracted))) \label{eq:contracted}
    \end{align}

  {\em 2. Retraction.} Next, we transform $A_M$ into a new algorithm
  $R_{A_M}$ called the {\em retraction} of $A_M$, as follows:

  (a) During round 1 of $R_{A_M}$, each input server for $S_j$ sends
  (in addition to the messages sent by $A_M$) every tuple in $t \in
  S_j$ to all servers $u$ that eventually receive $t$.  In other
  words, the input server sends $t$ to every $u$ for which there
  exists $k \in [r+1]$ such that $t \in
  K^{A_M,u,k}_{m_{A_M,u,k}(I_M)}(S_j)$.  This is possible because of
  the restrictions in the tuple-based \mpc($\varepsilon$)
  model: all destinations of  $t$ depend only on $S_j$,
  and hence can be computed by the input server.  
  Note that this does not increase the total number of bits received by any
  processor, though it does shift more of those bits to the first round.
  $R_{A_M}$ will not
  send any atomic tuples during rounds $k \geq 2$.  (b) In round $2$,
  $R_{A_M}$ sends {\em no} tuples. (c) In rounds $k \geq 3$, $R_{A_M}$
  sends a tuple $t$ from $u$ to $v$ if server $u$ knows $t$ at round
  $k$, and algorithm $A_M$ sends $t$ from $u$ to $v$ at round $k$.

  It follows that, for each round $k$, and for each subquery $q'$ of
  $q/\contracted$ with at least two atoms, $K^{R_{A_M},u,k}_{m(i)}(q')
  \subseteq K^{A_M,u,k}_{m(i)}(q')$: in other words, $R_{A_M}$ knows a
  subset of the non-atomic tuples known by $A_M$.  Moreover, let
  $J_+^{A_M}(i_M)$ be the set of non-atomic tuples known by $A_M$
  after round 1, $J_+^{A_M}(i_M) = \bigcup
  \setof{K^{R_{A_M},u,1}_{m(i)}(q')}{q' \mbox{ has at least two
      atoms}}$: these are the tuples that we refused to sent in round
  2.  Then:
  \begin{align}
    A_M(i_M) \subseteq (q/\contracted(i_M) \ltimes J_+^{A_M}) \cup  R_{A_M}(i_M) \label{eq:retract}
  \end{align}

\begin{sloppypar}
  Since $R_{A_M}$ wastes one round, we can compress it to an algorithm
  $A'$ with only $r$ rounds.  To prove the lemma, we convert
  \eqref{eq:retract} into a statement about $A$.
  \eqref{eq:contracted} already showed that $A_M(i_M)$ is related to
  $A(i_M, i_\contracted)$.  Now we show how $J_+^{A_M}$ is related to
  $J_\varepsilon^{A,q}(i)$:
  $J_+^{A_M}(i_M) \subseteq  \sigma^{-1}(\Pi_{\text{vars}(q/\contracted)}(J_\varepsilon^{A,q}(\bar  \sigma(i))))$
  because, by the definition of $\varepsilon$-goodness, if a subquery
  $q'$ of $q$ has two atoms in $M$, then $q' \not\in
  \Gamma^1_\varepsilon$.  \eqref{eq:retract} becomes:
\end{sloppypar}
  \begin{align*}
    A_M(i_M) \subseteq (q/\contracted(i_M) \ltimes  \Pi_{\text{vars}(q/\contracted)}(J_\varepsilon^{A,q}(i))) \cup \bar \sigma^{-1}(A'(i_M))
  \end{align*}
  The lemma follows from 
  $$q/\contracted(i_M) \ltimes
  \Pi_{\text{vars}(q/\contracted)}(J_\varepsilon^{A,q}(i)) \subseteq
  \Pi_{\text{vars}(q/\contracted)}(q(i) \ltimes
  J_\varepsilon^{A,q}(i))$$ 
  and $|A_M(i_M)| = |A(i_M, i_\contracted)|$,
  by \eqref{eq:contracted}.
\end{proof}

\begin{lemma}
\label{lemma:onebyp}
Let $q$ be a conjunctive query, and $q'$ a subquery; if $i$ is a
database instance for $q$, we write $i'$ for its restriction to the
relations occurring in $q'$.  Let $B$ be any algorithm for $q'$
(meaning that, for every matching database $i'$, $B(i') \subseteq
q'(i')$), and assume that $\E[|B(I')|] \leq \gamma \cdot
\E[|q'(I')|]$.  Then,
$\E[|q(I) \ltimes B(I')|] \leq \gamma \E[|q(I)|]$
where $I$ is a uniformly chosen matching database.
\end{lemma}
While, in general, $q'$ may return many more answers than $q$, the
lemma says that, if $B$ returns only a fraction of $q'$, then $q
\ltimes B$ returns only the same fraction of $q$.
\begin{proof}
  Let $\bar y = (y_1, \ldots, y_k)$ be the variables occurring in
  $q'$.  For any $\bar a \in [n]^k$, let $\sigma_{\bar y=\bar
    a}(q(i))$ denote the subset of tuples $t \in q(i)$ whose
  projection on $\bar y$ equals $\bar a$.  By symmetry, the quantity
  $\E[|\sigma_{\bar y = \bar a}(q(I))|]$ is independent of $\bar a$,
  and therefore equals $\E[|q(I)|]/n^k$.  Notice that $\sigma_{\bar
    y=\bar a}(B(i'))$ is either $\emptyset$ or $\set{\bar a}$.  We
  have:
  \begin{align*}
    \E[|q(I) \ltimes B(I')|] & = \sum_{\bar a \in [n]^k} \E[|\sigma_{\bar y = \bar a}(q(I)) \ltimes \sigma_{\bar y=\bar a}(B(I'))|]\\
& =  \sum_{\bar a \in [n]^k}\E[|\sigma_{\bar y = \bar a}(q(I))|] \cdot \P(\bar a \in B(I'))\\
& =  \E[|q(I)|] \cdot \sum_{\bar a \in [n]^k}\P(\bar a \in B(I'))/n^k \\
& = \E[|q(I)|]\cdot \E[|B(I')|]/n^k
  \end{align*}
  Repeating the same calculations for $q'$ instead of $B$,
  \begin{align*}
  \E[|q(I) \ltimes q'(I')|] = & \E[|q(I)|] \E[|q'(I')|]/n^k
\end{align*}
The lemma follows immediately, by using the fact that, by definition,
$q(i) \ltimes q'(i') = q(i)$.
\end{proof}

Finally, we prove \autoref{th:strong-multiround}.

\begin{proof}[Proof of \autoref{th:strong-multiround}]
  Given the $(\varepsilon,r)$-plan $\text{atoms}(q)$
  $=M_0\supset\ldots\supset M_r$, define $\hat
  M_k=\contracted_k-\contracted_{k-1}$, for $k \ge 1$.  We build up
  $i_{\contracted_r}$ by iteratively choosing matching databases
  $i_{\hat M_k}=\contracted_k-\contracted_{k-1}$ for $k=1,\ldots,
  r$ and applying \autoref{lemma:contraction} with $q$ replaced by
  $q/\contracted_{k-1}$ and $M$ replaced by $M_k$ to obtain algorithms
  $A^k=A^k_{(i_{\hat M_1},\ldots,i_{\hat M_k})}$ for $q/\hat
  M_1\cdots\hat M_k$ such that the following inequality holds for
  every choice of matching databases given by $i_{M_r}$ and
  $i_{\contracted_r}=(i_{\hat M_1},\ldots,i_{\hat M_r})$:
\begin{align}
|A(i_{M_r},i_{\contracted_r})|\nonumber
  & = |A(i_{M_r}, i_{\hat M_1}, \ldots, i_{\hat M_r})| \nonumber\\
    & \leq |q(i_{M_r},i_{\contracted_r}) \ltimes J^{A,q}_\varepsilon(i_{M_r}, i_{\hat M_1}, \ldots, i_{\hat M_r})| \nonumber\\
  &+ |q(i_{M_r},i_{\contracted_r}) \ltimes J^{A^1,q/\hat M_1}_\varepsilon(i_{M_r}, i_{\hat M_2}, \ldots, i_{\hat M_r})| \nonumber\\
  & + \ldots \nonumber\\
  & + |q(i_{M_r},i_{\contracted_r}) \ltimes J^{A^{r-1},q/\hat M_1\cdots\hat M_{r-1}}_\varepsilon(i_{M_r}, i_{\hat M_r})| \nonumber\\
  & +|A^r(i_{M_r})|\label{eq:sum}
\end{align}
We now average~\eqref{eq:sum} over a uniformly chosen matching database $I$
and upper bound each of the resulting terms:
For all $k \in[r]$ we have $\chi(q/\contracted_k)=\chi(q)$ (see
\autoref{def:m}), and hence, by \autoref{lem:expected_size}, we have
$\E[|q(I)|]=\E[|(q/\contracted_k)(I_{M_k})|]$.
By definition, we have $\tau^*(q/\contracted_r)\ge \tau^*(\mathcal{M})$.
Then, by \autoref{th:onestep}, \autoref{prop:friedgut}, and the fact that the
number of bits/tuples received by each processor in the first round of
algorithm $A^r$ is at most $r+1$ times the bound for the original algorithm $A$,
\begin{align*}
\E[|A^r(I_{M_r})|] 
\leq g_{q/\contracted_r,c(r+1)} \frac{\E[|(q/\contracted_r)(I_{M_r})|]}{p^{\tau^*(q/\contracted_r)(1-\varepsilon)-1}} 
 \leq g_{q/\contracted_r,c(r+1)} \frac{\E[|q(I)|]}{p^{\tau^*(\mathcal{M})(1-\varepsilon)-1}}
\end{align*}
\begin{sloppypar}
Note that $I_{M_{k-1}}=(I_{M_r},I_{\hat M_k},\ldots, I_{\hat M_{r}})$
and consider the expected number of tuples in $J = J^{A^{k-1},q/\hat
  M_1\cdots\hat M_{k-1}}_\varepsilon(I_{M_{k-1}})$.  The algorithm
$A^{k-1}=A^{k-1}_{I_{\contracted_{k-1}}}$ itself depends on the choice
of $I_{\contracted_{k-1}}$; still, we show that $J$ has a small number
of tuples.  Every subquery $q'$ of $q/\hat M_1\cdots\hat M_{k-1}$ that
is not in $\Gamma^1_\varepsilon$ (hence contributes to $J$) has
$\tau^*(q')\ge \tau^*(\mathcal{M})$.  By \autoref{th:onestep}, 
\autoref{prop:friedgut}, 
for
each fixing $I_{\contracted_{k-1}}=i_{\contracted_{k-1}}$, the
expected number of tuples produced for subquery $q'$ by $B_{q'}$ ,
where $B_{q'}$ is the portion of the first round of
$A^{k-1}_{i_{\contracted_{k-1}}}$ that produces tuples for $q'$,
satisfies $\E[|B_{q'}(I_{M_{k-1}})|]
\leq g_{q',c(r+1)} \E[|q'(I_{M_{k-1}})|]/p^{\tau^*(\mathcal{M})(1-\varepsilon)-1}$ since each processor in a round of $A^{k-1}_{i_{\contracted_{k-1}}}$  (and
hence $B_{q'}$) receives at most $r+1$ times the communication bound for 
a round of $A$.
We now apply \autoref{lemma:onebyp} to derive 
\end{sloppypar}
\begin{align*}
\E[|  q(I)\ltimes B_{q'}(I_{M_{k-1}})|]
 & = \E[|(q/\contracted_{k-1})(I_{M_{k-1}})\ltimes B_{q'}(I_{M_{k-1}})|]\\
&\leq g_{q',c(r+1)} (\E[|(q/\contracted_{k-1})(I_{M_{k-1}})|]/p^{\tau^*(\mathcal{M})(1-\varepsilon)-1}\\
& = g_{q',c(r+1)} (\E[|q(I)|]/p^{\tau^*(\mathcal{M})(1-\varepsilon)-1}.
\end{align*}
Averaging over all choices of $I_{\contracted_{k-1}}=i_{\contracted_{k-1}}$
and  summing over the number of different queries $q'$  in 
$\mathcal{S}(q/\hat M_1\cdots\hat M_{k-1}) = \mathcal{S}(q/\contracted_{k-1})$,
where we recall that  $\mathcal{S}(q/\contracted_{k-1})$ is the set of all
minimal connected subqueries $q'$ of $q/\contracted_{k-1}$ that are not
in $\Gamma^1_\varepsilon$,
we obtain 
\begin{align*}
\E[|q(I) \ltimes &J^{A^{k-1},q/\hat M_1\cdots\hat M_{k-1}}_\varepsilon(I_{M_{k-1}})|]
\leq \left( \sum_{q' \in \mathcal{S}(q/\contracted_{k-1})} g_{q',c(r+1)} \right)
\frac{\E[|q(I)|]}{p^{\tau^*(\mathcal{M})(1-\varepsilon)-1}}.
\end{align*}
Combining the bounds obtained for the $r+1$ terms in~\eqref{eq:sum}, we conclude
that
\begin{align*}
\E[|A(I)|]&\leq \left(g_{q/\contracted_r,c(r+1)} +
\sum_{k=1}^{r} \sum_{q' \in \mathcal{S}(q/\contracted_{k-1})} g_{q',c(r+1)} \right)
\frac{\E[|q(I)|]}{p^{\tau^*(\mathcal{M})(1-\varepsilon)-1}}\\
&=\gamma_c(q,\mathcal{M})\cdot \frac{\E[|q(I)|]}{p^{\tau^*(\mathcal{M})(1-\varepsilon)-1}}
\end{align*}
which proves \autoref{th:strong-multiround}.
\end{proof}

We now can apply the explicit bounds of \autoref{th:strong-multiround} to
prove~\autoref{th:connected-comps}.

\begin{proof}[Proof of \autoref{th:connected-comps}]
  Since larger $\varepsilon$ implies a more powerful algorithm, we
  assume without loss of generality that $\varepsilon=1-1/t$ for some
  integer constant $t\ge 1$.
  Let $\delta=1/(2t)$.
  The family of input graphs and the initial
  distribution of the edges to servers will look like an input to
  $L_k$, where $k =\lfloor p^\delta\rfloor$.
  In particular, the $n$ vertices of the input
  graph $G$ will be partitioned into $k+1$ sets $P_1, \dots, P_{k+1}$, 
  each partition containing $n/(k+1)$ vertices. The edges of $G$
  will form permutations between adjacent partitions, $P_i, P_{i+1}$,
  for $i=1, \dots, k$. Thus, $G$ will contain $nk/(k+1) < n$ edges. This
  construction creates essentially $k$ binary relations, each of size
  $n/(k+1)$. 
  
  Since $k<p$, we can assume that the adversary initially places the edges
  of the graph so that each server is given edges only from one relation. 
  It is now
  easy to see that any tuple-based algorithm in \mpc($\varepsilon$) that solves
  \textsc{Connected-Components} for an arbitrary graph $G$ of the above family
  in $r$ rounds implies an $(r+1)$-round tuple-based algorithm in
  \mpc($\varepsilon$) that solves 
  $L_k$ when each relation has size $n/(k+1)$ and $k = p$. Indeed, the new
  algorithm runs the algorithm for connected components for the first $r$
  rounds, and
  then executes a join on the labels of each node. 
  Since each tuple in $L_k$ corresponds
  exactly to a connected component in $G$, the join will recover all
  the tuples of $L_k$.
  
  Since the query size is not independent of the number of servers $p$,
  we have to carefully compute the constants for our lower bounds.
  To conclude our proof, consider an \mpc($\varepsilon,c$) algorithm
  for $L_k$.
  Let $r=\lceil \log_{k_{\varepsilon}} k \rceil -1$.
  We will use the $(\varepsilon,r)$-plan $\mathcal{M}$ for $L_k$
  presented in the proof of \autoref{lemma:lk:lower}, apply
  \autoref{th:strong-multiround},
  and compute the factor $\gamma_c{L_k,\mathcal{M}}$.
  First, notice that each query $L_k/\contracted_j$ for $j=0, \dots, r$ is
  isomorphic to $L_{k/k_{\varepsilon}^j}$.
  Then, the set $\mathcal{S}(L_{k/k_{\varepsilon}^j})$
  consists of at most $k/k_{\varepsilon}^j$ paths $q'$ of
  length $k_{\varepsilon}+1$.
  Observe that $k_\varepsilon=2t$ since $\varepsilon=1-1/t$.
  Also, by the choice of $r$,
  $L_k/\contracted_r$ is isomorphic to $L_\ell$ where $\ell\ge k_\varepsilon+1$
  and $\ell< k_\varepsilon^2$.
  Since $L_{k'}$ has total arity $2k'$, $k'$ atoms, and
  $\tau^*(L_{k'})=\lceil k'/2\rceil$, we derive that
  $g_{L_{k'},c}=(2k'c/\lceil k'/2\rceil)^{\lceil k'/2\rceil}
  \le (4c)^{\lceil k'/2\rceil}$.
  Thus, we have
  \begin{align*}
\gamma_c(L_k,\mathcal{M}) 
&= g_{L_k/\contracted_r,c(r+1)} +
\sum_{j=1}^{r} \sum_{q' \in \mathcal{S}(L_k/\contracted_{j-1})} g_{q',c(r+1)} \\
&\le (4c(r+1))^{\lceil k_{\varepsilon}^2/2\rceil} + 
\sum_{j=1}^{r} \frac{k}{k_{\varepsilon}^{j-1}} (4c(r+1))^{\lceil (k_{\varepsilon}+1)/2\rceil}\\
&\leq (2k+1) (4c(r+1))^{\lceil k_{\varepsilon}^2/2\rceil}\\
&        \leq (2k+1) (4c\lceil\log_{k_\varepsilon} k\rceil)^{\lceil k_{\varepsilon}^2/2\rceil}.
\end{align*}
In particular this implies that $\gamma_c(L_k,\mathcal{M})$ is at most
$c' k \cdot (\log_2 k)^{c''}$ for some constants $c'$ and $c''$ depending
only on $\varepsilon$ and $c$.
Consequently, \autoref{th:strong-multiround} implies that any tuple-based 
\mpc($\varepsilon$) algorithm 
using at most $\lceil \log_{k_{\varepsilon}} k \rceil-1$ rounds 
  reports at most a
\begin{align*}
\frac{c' k \cdot (\log_2 k)^{c''}}
  {p^{\tau^*(\mathcal{M})(1-\varepsilon)-1}}
\le c' p^{(1+\delta)-\tau^*(\mathcal{M})(1-\varepsilon)} (\delta\log_2 p)^{c''}
\end{align*}
fraction of the $n/(k+1)$ required output tuples for the $L_k$ query.
Now by construction,
$\tau^*(\mathcal{M})=\tau^*(K_{k_\varepsilon+1})=\lceil (k_\varepsilon+1)/2\rceil=t+1$
since $k_\varepsilon=2t$. 
Now since $1-\varepsilon=1/t$ and $\delta=1/(2t)$, we see that the fraction
of required tuples reported is at most
$$c' p^{(1+1/(2t))-(t+1)/t)} (\delta\log_2 p)^{c''}\le c'p^{-1/(2t)} (\log_2 p)^{c''}$$
which is $o(1)$ in $p$ since $t\ge 1$, $c'$, $c''$ are constants.

This implies that any algorithm that computes \textsc{Connected-Components} on $G$ requires at least  $\lceil \log_{k_\varepsilon} \lfloor p^\delta\rfloor \rceil -2 = \Omega(\log p)$ rounds,
  since $k_\varepsilon$ and $\delta$ are constants.
\end{proof}

%% file: conclusion.tex
\section{Conclusion}
\label{sec:conclusion}

We have introduced powerful models for capturing tradeoffs between
rounds and amount of communication required for
parallel computation of relational queries.  For one round 
on the most general model we have shown that queries are characterized by
$\tau^*$ which determines the space exponent
$\varepsilon=1-1/\tau^*$ that governs the replication rate as a function of
the number of processors.
For multiple rounds we derived a strong lower bound
tradeoff between the number of rounds
$r$ and the replication rate of $r \cdot \log 2/(1-\varepsilon) \approx \log
(\text{rad}(q))$ for more restricted tuple-based communication.
For both, we showed matching or nearly matching upper bounds given by 
simple and  natural algorithms. 

\cut{
In this paper, we explored the tradeoff between the amount of
communication and rounds for parallel computation of relational
queries. For one round, we gave matching upper and lower bounds for
the amount of communication needed, by showing that the covering number 
$\tau^*$ of a query and its space exponent $\varepsilon$ are related as 
$\varepsilon = 1-1/\tau^*$.  For multiple rounds, we
showed lower bounds for a restricted type of tuple-based
communication, and provided matching and near-matching bounds for
several classes of queries.  The tradeoff between the number of rounds
$r$ and the replication rate (expressed through the space exponent
$\varepsilon$) is $r \cdot \log 2/(1-\varepsilon) \approx \log
(\text{rad}(q))$.  Future work includes a study of multi-round lower
bounds for an unrestricted \mpc($\varepsilon$) model, as well as
extensions to problems beyond conjunctive queries.
}

%% file: paper.bbl
\begin{thebibliography}{10}

\bibitem{DBLP:journals/corr/abs-1206-4377}
F.~N. Afrati, A.~D. Sarma, S.~Salihoglu, and J.~D. Ullman.
\newblock Upper and lower bounds on the cost of a map-reduce computation.
\newblock {\em CoRR}, abs/1206.4377, 2012.

\bibitem{DBLP:conf/edbt/AfratiU10}
F.~N. Afrati and J.~D. Ullman.
\newblock Optimizing joins in a map-reduce environment.
\newblock In {\em EDBT}, pages 99--110, 2010.

\bibitem{ams:freq}
N.~Alon, Y.~Matias, and M.~Szegedy.
\newblock The space complexity of approximating the frequency moments.
\newblock {\em JCSS}, 58(1):137--147, 1999.

\bibitem{DBLP:conf/focs/AtseriasGM08}
A.~Atserias, M.~Grohe, and D.~Marx.
\newblock Size bounds and query plans for relational joins.
\newblock In {\em FOCS}, pages 739--748, 2008.

\bibitem{DBLP:conf/pods/Chaudhuri12}
S.~Chaudhuri.
\newblock What next?: a half-dozen data management research goals for big data
  and the cloud.
\newblock In {\em PODS}, pages 1--4, 2012.

\bibitem{DBLP:journals/siamdm/ChungFGG88}
F.~R.~K. Chung, Z.~F{\"u}redi, M.~R. Garey, and R.~L. Graham.
\newblock On the fractional covering number of hypergraphs.
\newblock {\em SIAM J. Discrete Math.}, 1(1):45--49, 1988.

\bibitem{DBLP:conf/osdi/DeanG04}
J.~Dean and S.~Ghemawat.
\newblock Mapreduce: Simplified data processing on large clusters.
\newblock In {\em OSDI}, pages 137--150, 2004.

\bibitem{datasciencesurvey}
{EMC Corporation}.
\newblock Data science revealed: A data-driven glimpse into the burgeoning new
  field.
\newblock
  \url{http://www.emc.com/collateral/about/news/emc-data-science-study-wp.pdf}.

\bibitem{DBLP:journals/talg/FeldmanMSSS10}
J.~Feldman, S.~Muthukrishnan, A.~Sidiropoulos, C.~Stein, and Z.~Svitkina.
\newblock On distributing symmetric streaming computations.
\newblock {\em ACM Transactions on Algorithms}, 6(4), 2010.

\bibitem{friedgut2004hypergraphs}
E.~Friedgut.
\newblock Hypergraphs, entropy, and inequalities.
\newblock {\em American Mathematical Monthly}, pages 749--760, 2004.

\bibitem{DBLP:conf/focs/GalG07}
A.~G{\'a}l and P.~Gopalan.
\newblock Lower bounds on streaming algorithms for approximating the length of
  the longest increasing subsequence.
\newblock In {\em FOCS}, pages 294--304, 2007.

\bibitem{DBLP:journals/jlp/GangulyST92}
S.~Ganguly, A.~Silberschatz, and S.~Tsur.
\newblock Parallel bottom-up processing of datalog queries.
\newblock {\em J. Log. Program.}, 14(1{\&}2):101--126, 1992.

\bibitem{DBLP:conf/soda/GroheM06}
M.~Grohe and D.~Marx.
\newblock Constraint solving via fractional edge covers.
\newblock In {\em SODA}, pages 289--298, 2006.

\bibitem{DBLP:conf/icalp/GuhaH09}
S.~Guha and Z.~Huang.
\newblock Revisiting the direct sum theorem and space lower bounds in random
  order streams.
\newblock In {\em ICALP}, volume 5555 of {\em LNCS}, pages 513--524. Springer,
  2009.

\bibitem{hadoop}
Hadoop.
\newblock \url{http://hadoop.apache.org/}.

\bibitem{DBLP:conf/soda/KarloffSV10}
H.~J. Karloff, S.~Suri, and S.~Vassilvitskii.
\newblock A model of computation for mapreduce.
\newblock In {\em SODA}, pages 938--948, 2010.

\bibitem{DBLP:conf/pods/KoutrisS11}
P.~Koutris and D.~Suciu.
\newblock Parallel evaluation of conjunctive queries.
\newblock In {\em PODS}, pages 223--234, 2011.

\bibitem{kn97}
E.~Kushilevitz and N.~Nisan.
\newblock {\em Communication Complexity}.
\newblock Cambridge University Press, Cambridge, England ; New York, 1997.

\bibitem{DBLP:journals/pvldb/MelnikGLRSTV10}
S.~Melnik, A.~Gubarev, J.~J. Long, G.~Romer, S.~Shivakumar, M.~Tolton, and
  T.~Vassilakis.
\newblock Dremel: Interactive analysis of web-scale datasets.
\newblock {\em PVLDB}, 3(1):330--339, 2010.

\bibitem{DBLP:conf/pods/NgoPRR12}
H.~Q. Ngo, E.~Porat, C.~R{\'e}, and A.~Rudra.
\newblock Worst-case optimal join algorithms: [extended abstract].
\newblock In {\em PODS}, pages 37--48, 2012.

\bibitem{DBLP:conf/sigmod/OlstonRSKT08}
C.~Olston, B.~Reed, U.~Srivastava, R.~Kumar, and A.~Tomkins.
\newblock Pig latin: a not-so-foreign language for data processing.
\newblock In {\em SIGMOD Conference}, pages 1099--1110, 2008.

\bibitem{DBLP:conf/www/SuriV11}
S.~Suri and S.~Vassilvitskii.
\newblock Counting triangles and the curse of the last reducer.
\newblock In {\em WWW}, pages 607--614, 2011.

\bibitem{TSJSCALWM09}
A.~Thusoo, J.~S. Sarma, N.~Jain, Z.~Shao, P.~Chakka, S.~Anthony, H.~Liu,
  P.~Wyckoff, and R.~Murthy.
\newblock Hive - a warehousing solution over a map-reduce framework.
\newblock {\em PVLDB}, 2(2):1626--1629, 2009.

\bibitem{tiw87}
P.~Tiwari.
\newblock Lower bounds on communication complexity in distributed computer
  networks.
\newblock {\em JACM}, 34(4):921--938, Oct. 1987.

\bibitem{DBLP:journals/crossroads/Ullman12}
J.~D. Ullman.
\newblock Designing good mapreduce algorithms.
\newblock {\em ACM Crossroads}, 19(1):30--34, 2012.

\bibitem{yao83}
A.~C. Yao.
\newblock Lower bounds by probabilistic arguments.
\newblock In {\em FOCS}, pages 420--428, Tucson, AZ, 1983.

\end{thebibliography}
